\documentclass[aps,prl,twocolumn,groupedaddress,appendix]{revtex4-2}
\usepackage{amsmath}
\usepackage{amssymb}
\usepackage{booktabs} 
\usepackage{array}    
\usepackage{siunitx}  
\usepackage{bm}
\usepackage{bbm}
\usepackage{braket}
\usepackage{amsthm}
\usepackage{diagbox}
\usepackage{setspace}
\usepackage{verbatim}
\usepackage{esint}
\usepackage{physics}
\usepackage{indentfirst}
\usepackage{graphicx}
\usepackage{pdfpages}
\makeatletter
\AtBeginDocument{\let\LS@rot\@undefined}  
\makeatother
\usepackage{float}
\usepackage{dsfont}
\usepackage{mathrsfs}
\usepackage{mathtools}
\usepackage{subfiles}
\usepackage{tikz-cd}
\usepackage{simplewick}
\usepackage{slashed}
\usepackage{mathtools}
\usepackage[colorlinks,linkcolor=blue,citecolor=blue,anchorcolor=blue]{hyperref}
\usepackage{xspace}
\newcommand{\pkg}[1]{\textsf{#1}\xspace}
\usepackage{slashed}
\usepackage{cancel}
\usepackage{xcolor}
\usepackage[normalem]{ulem}
\usepackage{youngtab}
\usepackage{ytableau}
\usepackage{enumerate}
\usepackage{leftindex}
\usepackage{svg}
\usepackage[T5]{fontenc}
\usepackage{cancel}
\usepackage{pst-node}

\newtheorem{theorem}{Theorem}

\usepackage{booktabs}
\usepackage{siunitx} 
\usepackage{threeparttable} 

\let \oldbm \bm
\renewcommand{\vec}[1]{\oldbm{#1}}

\usepackage{amsthm}

\theoremstyle{definition}

\newtheorem{definition}{Definition}

\theoremstyle{remark}
\newtheorem{remark}{Remark}
\newtheorem*{remark*}{Remark}

\def\bk{{\vec k}}

\def\bA{{\vec A}}
\def\bB{{\vec B}}

\def\bQ{{\vec Q}}

\def\bB{{\bf B}}

\def\bm{{\vec m}}

\def\bsigma{{\boldsymbol \sigma}}

\def\Tr{\mathop{\mathrm{Tr}}}

\newcommand{\beq}{\begin{equation}}
\newcommand{\eeq}{\end{equation}}
\newcommand{\beqarray}{\begin{eqnarray}}
\newcommand{\eeqarray}{\end{eqnarray}}

\hypersetup{colorlinks=true,
            citecolor=blue,
            linkcolor=blue,
            urlcolor=blue,
            bookmarksnumbered=true}

\begin{document}

\definecolor{lightblue}{HTML}{1A73E8} 

\title{Bootstrapping Flat-band Superconductors:\\ Rigorous Lower Bounds on Superfluid Stiffness}

\author{Qiang~Gao}
\author{Zhaoyu~Han}
\author{Eslam~Khalaf}
\affiliation{Department of Physics, Harvard University, Cambridge, Massachusetts 02138, USA}

\date{\today}

\begin{abstract}
The superfluid stiffness fundamentally constrains the transition temperature of superconductors, especially in the strongly coupled regime. However, accurately determining this inherently quantum many-body property in microscopic models remains a significant challenge. In this work, we show how the \textit{quantum many-body bootstrap} framework, specifically the reduced density matrix (RDM) bootstrap, can be leveraged to obtain rigorous lower bounds on the superfluid stiffness in frustration free interacting models with superconducting ground state. We numerically apply the method to a special class of frustration free models, which are known as quantum geometric nesting models, for flat-band superconductivity, where we uncover a general relation between the stiffness and the pair mass. Going beyond the familiar Hubbard case within this class, we find how additional interactions, notably simple 
magnetic couplings, can enhance the superfluid stiffness. Furthermore, we find that the RDM bootstrap unexpectedly reveals that the trion-type correlations are essential for bounding the stiffness, offering new insights on the structure of these models.
Straight-forward generalization of the method can lead to bounds on susceptibilities complementary to variational approaches.
Our findings underscore the immense potential of the quantum many-body bootstrap as a powerful tool to derive rigorous bounds on physical quantities beyond energy.
\end{abstract}

\maketitle

\textit{Introduction.--} Superfluid stiffness ($D_\text{s}$) is a central property of superconductors (SC)~\cite{PhysRevB.47.7995}, which limits the SC transition temperature $T_c$ along with the pairing gap. Particularly for two-dimensional SCs, it has units of energy and controls the Berzinksii-Kosterlitz-Thouless transition temperature \cite{NelsonPRL1977}. The realization of SC phases in flat band (FB) moir'e materials \cite{CaoNature2018, CaoNature2018a, yankowitz2019tuning, park2021tunable, hao2021electric, zhou2021superconductivity, han2025signatures} has highlighted the importance of understanding stiffness in FB systems. While in conventional SCs, $D_\text{s}$ is dominantly determined by single-particle dispersion, such contribution vanishes in FB systems. As a result, superfluid stiffness in these systems arises entirely from the structure of the FB wavefunctions, known as quantum geometry. This has motivated efforts to study superfluid stiffness in FB SC and derive bounds on its value \cite{EmeryNature1995, Paramekanti1998, Hazra2019, Verma2021, mao2023, Hofmann2022,ShenPRL1993, TormaNatPhy2015, Julku2016, XiePRL2020, Herzog2022, herzog2022many, PhysRevB.106.014518,PhysRevB.110.024507}. We also note recent works which related quantum geometry to to various other properties of the SC~\cite{PhysRevLett.132.026002,hu2025anomalous,PhysRevResearch.7.023273,li2025vortexstatescoherencelengths}.

However, obtaining this quantity from unbiased calculations of microscopic models has been particularly challenging, since this is an intrinsically quantum many-body quantity that in principle receives contributions from {\it all} excited states~\cite{Note_on_contribution_from_excited_states}. In fact, most studies of FB superconductivity either considered a mean-field treatment which is generally uncontrolled or computed stiffness from the pair mass extracted from the two-body spectrum. The latter only provides an upper bound on the true many-body stiffness since the true Cooper pairs may be significantly dressed by particle-hole excitations that make them heavier, leading to smaller stiffness. Under certain assumptions, the two-body pair mass can be related to the quantum metric, which is lower bounded by the Chern number. Thus, these topological bounds correspond to ``lower bounds on upper bounds'', which do not necessarily imply strict bounds on the many-body stiffness or even that it has to be non-zero.
To date, evaluating stiffness has only been done numerically within a few examples~\cite{PhysRevLett.127.025301,zhang2025optimizing,PhysRevB.102.201112,PhysRevLett.130.226001,Herzog2022,PhysRevB.105.024502,PhysRevB.94.245149} using density matrix renormalization group in quasi-one-dimensional systems, or quantum Monte Carlo (QMC) methods for sign-problem-free models.

Confronting such challenges calls for non-perturbative, scalable techniques. The quantum-many-body bootstrap has been recently proposed as an efficient non-perturbative method for quantum many-body problems \cite{massaccesi2021variational, SemidefiniteRelaxation, haim2020variational,baumgratz2012lower,Han2020QMB,Kull2024LowerBounds,Gao2024QHBootstrap,scheer2024hamiltonian} with promising results on lattice models like the Hubbard model \cite{Han2020QMB, scheer2024hamiltonian}, spin chains \cite{haim2020variational, Kull2024LowerBounds}, and the fractional quantum Hall problem \cite{Gao2024QHBootstrap}. Although the bootstrap has a decades-long history in \emph{S}-matrix theory \cite{Chew1961Smatrix,Homrich2019SmatrixBootstrap}, conformal field theory \cite{Poland2019CFTreview}, and quantum chemistry \cite{Coleman1963, PhysRevLett.108.263002}, its application to condensed-matter systems is comparatively new, with early efforts devoted to rigorous lower bounds on ground state (GS) energies \cite{Han2020QMB,Kull2024LowerBounds,Berenstein2023SDPbootstrap,massaccesi2021variational, SemidefiniteRelaxation, haim2020variational,baumgratz2012lower, Gao2024QHBootstrap, scheer2024hamiltonian}. Recent work has shown that bootstrap-based approaches can also provide rigorous bounds on observables beyond the energy \cite{Wang2024CertifyGS, Fawzi2024, cho2024coarse}.

In this work, we show how the quantum many-body bootstrap approach can be used to rigorously lower bound $D_\text{s}$ at zero temperature in a special class of superconducting models that are frustration-free (FF). A FF model is one where the Hamiltonian can be written as a sum of few-body interactions such that each term is minimized by the same GS. We show that the energy lower bound obtained in the bootstrap approach is exact at the FF point. This allows us to rigorously lower-bound $D_\text{s}$ by lower-bounding the GS energy in the presence of a small flat gauge connection (equivalently a twisted boundary condition). Our approach complements numerical and analytical variational approaches that can only provide upper bounds.

To showcase the power of this method, we numerically study a class of FF models for flat-band SC, called the quantum geometric nesting (QGN) models~\cite{PhysRevX.14.041004,SM}; familiar examples in this class include attractive Hubbard models~\cite{ TormaNatPhy2015,PhysRevB.94.245149,PhysRevB.106.014518,PhysRevLett.132.026002} with a special property called uniform pairing condition (UPC) on the FBs~\cite{PhysRevB.94.245149,herzog2022many}. Remarkably, we find that in all cases the computed lower bound on stiffness equals (up to small numerical errors) a rigorous upper bound obtained from a BCS-type variational ansatz (see S.M.~\cite{SM} for the proof of the upper bound)
\begin{equation}\label{lower_bound_formula}
D_\text{s} = \frac{N_\text{flat}}{2}\nu (1-\nu) m_\text{pair}^{-1}
\end{equation}
where $m_\text{pair}^{-1}$ is the inverse Cooper pair mass evaluated at the same system size (defined as the second derivative of the lowest spectrum with respect to the total momentum in the two-particle sector), and $N_\text{flat}$ is the number of FBs (counting spin degeneracy).
Based on our results, we conjecture Eq.~\eqref{lower_bound_formula} to be the exact expression for $D_\text{s}$ for QGN models, leaving a rigorous proof to future work. Note that this remarkably simple result relates a genuinely quantum many-body property to a few-body one, which is much easier to compute and interpret. It also implies that topological or geometric lower bounds on the pair mass translate to lower bounds on the many-body stiffness. Physically, it indicates that in QGN models the bare Cooper pairs do not get dressed by particle-hole excitations, which is generally not true.

\begin{figure}[t]
    \centering
    \includegraphics[width=0.84\linewidth]{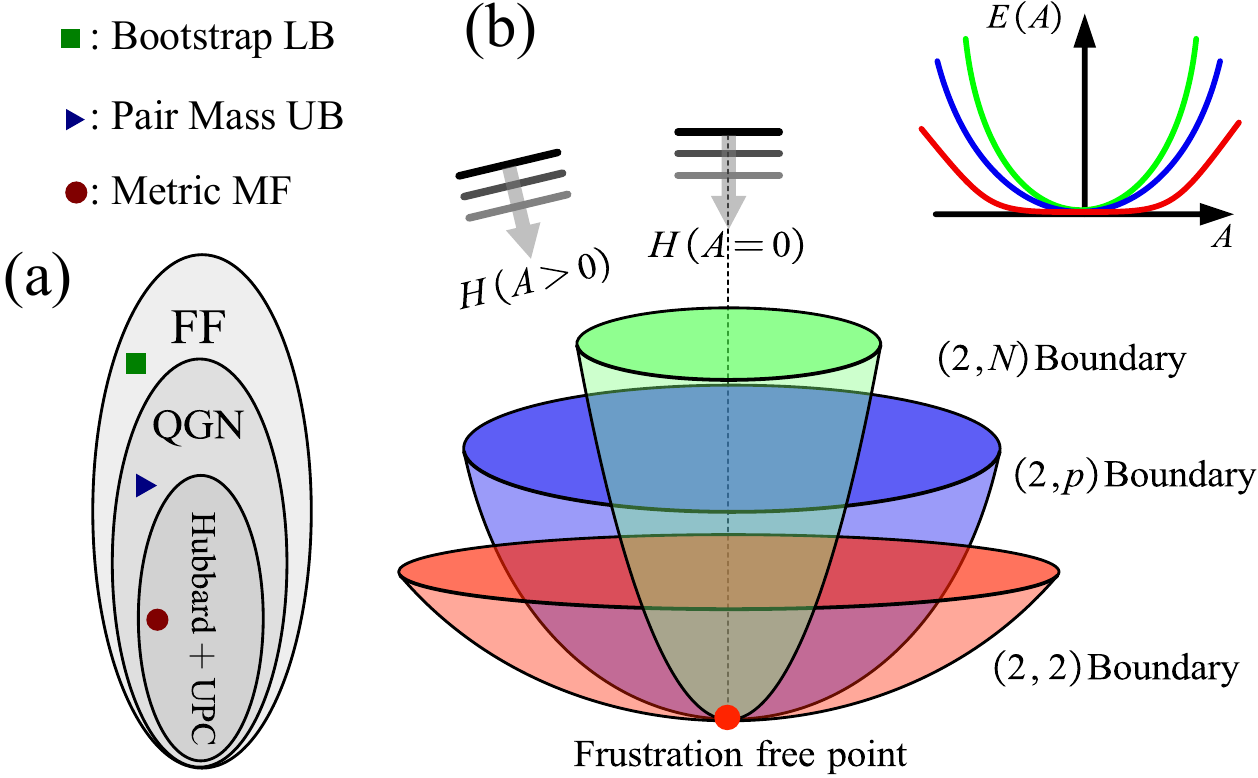}
     \caption{(a) Schematic illustration of classes of models with exact superconducting ground states. Our lower-bounding method applies to the most general class: frustration-free (FF) models. Within the FF set, quantum geometric nesting (QGN) models admit a variational upper bound based on the pair mass [Eq.~\eqref{lower_bound_formula}], which we prove in the SM and find to be saturated by the bootstrap lower bound, suggesting its exactness. A subset of QGN models with Hubbard interactions and uniform pairing condition (UPC) previously yielded mean-field results for the stiffness in terms of the minimal quantum metric~\cite{TormaNatPhy2015,PhysRevB.94.245149,PhysRevB.106.014518,herzog2022many}. (b) Schematic plot of feasible regions at different levels of the 2RDM bootstrap hierarchy. The red dot marks the frustration-free point where all feasible boundaries of $(2,p)$ constraints for $p=2,\cdots,N$ meet. The arrows indicate the optimization direction for different Hamiltonians (with different external parameter $A$). The upright corner schematically shows the energy $E(A)$ obtained after imposing $(2,p)$ constraints at different levels.
    } 
\label{fig:Bootstrap_Superconductor_schematic}
\end{figure}

In the Hubbard model and assuming UPC, the pair mass can be exactly evaluated as
$m_\text{pair}^{-1}=\frac{2g|U|}{N_\text{band}}$,
with $g$ the minimum quantum metric~\cite{PhysRevB.98.220511,herzog2022many}. Substituting this into Eq.~\eqref{lower_bound_formula} yields the previous mean-field result~\cite{TormaNatPhy2015,PhysRevB.94.245149,PhysRevB.106.014518}. Eq.~\eqref{lower_bound_formula} has thus been anticipated in the Hubbard case~\cite{herzog2022many,PhysRevB.109.174508,PhysRevA.99.053608,PhysRevB.98.220511}, and a series of topological bounds (on $g$ and thus $D_\text{s}$) based on this result has been put forward~\cite{XiePRL2020, Herzog2022}, which we now justify with a rigorous many-body lower bound. However, we stress that beyond the UPC Hubbard limit, $m_\text{pair}$ generically depends in a complicated way on the interaction and the FB wavefunctions $|u^{n}_\bk\rangle$. Intuitively, it measures the mismatch between the wavefunctions at paired momenta as a small center-of-mass momentum is introduced. The general relation Eq.~\eqref{lower_bound_formula} is thus a new result for general QGN models. 

Specifically, we performed numerical studies of the flat-band attractive Hubbard model with two different electronic structures: one with topological FBs~\cite{PhysRevB.102.201112} and the other with FBs with tunable quantum metric~\cite{PhysRevLett.130.226001}. Both models are sign-problem-free and were studied earlier using Monte Carlo~\cite{PhysRevB.102.201112, PhysRevLett.130.226001}. We then include interactions beyond the Hubbard limit such that the model remains in the QGN family (and thus FF) but not sign-problem-free~\cite{Note_on_sign_problem}, illustrating that our method goes beyond what is possible with QMC. Interestingly, we find a simple magnetic coupling that enhances the stiffness.

To date, there are few quantum many-body bootstrap schemes proposed, differing mainly in the bootstrap variables employed and the constraints imposed. In this work, we employ the reduced density matrix (RDM) bootstrap, which uses the two-particle RDM as the variable and imposes positivity constraints based on the hierarchy developed by Mazziotti~\cite{Maziotti2005, PhysRevLett.108.263002, Mazziotti2012, mazziotti2023quantum}. This method works for finite systems, but we show that our results follow simple finite-size scaling, allowing extrapolation to the thermodynamic limit. We emphasize that the RDM bootstrap goes beyond a purely numerical tool by directly providing information about particle correlations~\cite{Gao2024QHBootstrap}. In particular, we find that constraints built from trion operators ($c(c^\dagger c^\dagger)$ or $c^\dagger (cc)$) determine the stiffness.

\textit{2RDM bootstrap and constraint hierarchy.--} 
The bootstrap scheme we adopt relies on the following observation. 
For a two-body Hamiltonian $\hat H=\sum_{\alpha\beta\gamma\delta}\leftindex^2H^{\alpha\beta}_{\gamma\delta}c^\dagger_\alpha c^\dagger_\beta c_\delta c_\gamma$, it suffices to know the two-particle reduced density matrix (2RDM),
$\leftindex^2D^{\alpha \beta}_{\gamma \delta}\equiv\Tr[\rho c^\dagger_\alpha c^\dagger_\beta c_\delta c_\gamma]$,
with $\rho$ the GS density operator, to compute the GS energy \emph{exactly}: one minimizes the linear functional $\Tr{[\leftindex^2H\leftindex^2D]}$ over the space of 2RDMs without knowing the wavefunction (details can be found in S.M.~\cite{SM}).
The difficulty is that it is hard to determine whether a given two-particle density matrix (2DM) is reduced from a valid $N$-particle state. This is known as the $N$-representability problem~\cite{Coleman1963} and is QMA-hard (the quantum analogue of NP-hard)~\cite{schuch2009computational}. While computing the GS energy by minimizing over the exact set of 2RDMs (denoted $\leftindex^2_N{\mathbb{D}}$) is difficult, one can instead construct a superset $\leftindex^2_N{\tilde{\mathbb{D}}}\supseteq\leftindex^2_N{\mathbb{D}}$ satisfying constraints that are necessary but not sufficient for $N$-representability.
Minimizing the energy functional $E[\leftindex^2D]:= \Tr{[\leftindex^2H\leftindex^2D]}$ over $\leftindex^2_N{\tilde{\mathbb{D}}}$ therefore yields a rigorous lower bound on the GS energy, since the true GS 2RDM lies within $\leftindex^2_N{\mathbb{D}}$.

One important class of constraints is \textit{positivity}: for any operator $\hat{O}$, one must have
$\langle\hat{O}\hat{O}^\dagger\rangle\equiv\Tr[\rho\hat{O}\hat{O}^\dagger]\geq 0$.
The simplest example takes $\hat O_D = \sum_{ij} A^D_{ij} c_i^\dagger c_j^\dagger$.
Imposing $\langle \hat O_D \hat O_D^\dagger \rangle \geq 0$ for all $A^D_{ij}$ implies that the matrix
$\leftindex^2 D^{ij}_{kl} := \langle c^\dagger_i c^\dagger_j c_lc_k\rangle$,
i.e. the 2RDM, is positive semidefinite (PSD).
Two other choices are ${\hat{O}}_Q = \sum_{ij}A^Q_{ij}c_ic_j$ and
${\hat{O}}_G = \sum_{ij}A^G_{ij}c^\dagger_i c_j$, giving rise to the constraints
$\leftindex^2 G^{ij}_{kl} := \langle c^\dagger_i c_j c^\dagger_lc_k\rangle\succeq 0$
and
$\leftindex^2 Q^{ij}_{kl} := \langle c_ic_j c^\dagger_l c^\dagger_k\rangle\succeq 0$.
The matrices $\leftindex^2G$ and $\leftindex^2Q$ can be expressed in terms of $\leftindex^2D$ using fermionic anticommutation relations, yielding non-trivial constraints on the 2RDM.
The three operator classes above exhaust constraints generated by fermion bilinears.

To construct additional constraints, one must use higher-body operators, e.g. $\hat O \sim c^\dagger c c$.
However, the resulting positivity constraints $\langle \hat O \hat O^\dagger \rangle \geq 0$ cannot be expressed solely in terms of the 2RDM and would require introducing higher-body RDMs.
To circumvent this, one instead considers convex combinations of such constraints,
$\sum_i \langle \hat O_i \hat O_i^\dagger \rangle \geq 0$,
where the $\hat O_i$ are polynomials with $p>2$ fermion operators but whose sum is expressible entirely in terms of the 2RDM.
This defines the $(2,p)$ constraint hierarchy developed by Mazziotti~\cite{Mazziotti2012}.
The simplest case takes
$\hat O_{\rm T1} = \sum_{ijk} A^{\rm T1}_{ijk} c_i^\dagger c_j^\dagger c_k^\dagger$.
One then observes that the operator
$\hat O_{\rm T1} \hat O_{\rm T1}^\dagger + \hat O_{\rm T1}^\dagger\hat O_{\rm T1}$
contains only two-body terms due to fermionic antisymmetry, despite being constructed from three-body operators.
This yields the constraint
$\leftindex^3[T_1]^{ijk}_{lmn} := \langle c^\dagger_i c^\dagger_j c^\dagger_kc_nc_mc_l
+ c_lc_m c_n c^\dagger_k c^\dagger_j c^\dagger_i\rangle\succeq 0$,
where $\leftindex^3T_1$ depends only on the 2RDM.
Another constraint constructed from three-body operators is T2, corresponding to
$\leftindex^3[T_2]^{ijk}_{lmn}
= \langle c^\dagger_i c^\dagger_j c_k c^\dagger_nc_mc_l
+ c_lc_m c^\dagger_n c_k c^\dagger_j c^\dagger_i\rangle\succeq 0$.
Finally, Mazziotti~\cite{PhysRevLett.108.263002} showed that taking $p$ to the particle number $N$ in the $(2,p)$ hierarchy causes the feasible set $\leftindex^2_N{\mathbb{D}}^{(2,p)}$ to converge to the $N$-representable set:
\begin{equation}\label{2p_hierarchy}
    \begin{split}
    \leftindex^2_N{\mathbb{D}}^{(2,2)}\supseteq\leftindex^2_N{\mathbb{D}}^{(2,3)}
    \supseteq\cdots\supseteq \leftindex^2_N{\mathbb{D}}^{(2,N)}
    = \leftindex^2_N{\mathbb{D}}\\
   \Rightarrow E^{(2,2)}\leq E^{(2,3)}\leq \cdots\leq E^{(2,N)} = E_g.
    \end{split}
\end{equation}

\textit{Rigorous lower bound on the superfluid stiffness.--}
Superfluid stiffness is defined as the second derivative of the GS energy as a function of a flat connection (boundary twist): $\bk \mapsto \bk + \bA$. We note that having rigorous lower bounds on the GS energy generally does not provide sharp statements about its derivatives with respect to external parameters. However, if the GS energy at $\bA = 0$ is known exactly and the lower bounds are also exact at this point, then one can obtain rigorous bounds on the derivatives. Our argument to lower bound the stiffness proceeds as follows.

The RDM bootstrap provides a series of lower bounds on the GS energy for any flat connection $\bA$ from the $(2,p)$ hierarchy~\eqref{2p_hierarchy}:
\begin{equation}
    0\leq E^{(2,2)}(\bA)\leq\cdots\leq E^{(2,p)}(\bA)\leq \cdots\leq E_g(\bA),
\end{equation}
where we set the vacuum to be the zero energy reference.
If the Hamiltonian is FF, we know that $E_g(\bA = 0) = 0$. Furthermore, if the FF decomposition of the Hamiltonian involves only $p$-body operators, we show in End Matter that bootstrap bounds are exact with $(2,q)$ constraints for any $q \geq p$. This implies
$E^{(2,p)}(\bA = 0) =  E^{(2,p+1)}(\bA = 0) = \dots = E_g(\bA = 0) = 0$.
Most of the models we study are already FF at the two-body level, since the Hamiltonian is a sum of two-body interactions that annihilate the GS. This means that we obtain the exact GS energy already at the level of the $(2,2)$ constraints.
Finally, the global $U(1)$ gauge symmetry enforces
$\frac{\partial E_g}{\partial A}\big|_{A=0} = 0$,
which in turn implies
$\frac{\partial E^{(2,p)}}{\partial A}\big|_{A=0}
= \frac{\partial E^{(2,p+1)}}{\partial A}\big|_{A=0}
= \dots = 0$.
These considerations imply
\begin{equation}\label{lower_bound_stiffness}
    0\leq \left.\frac{\partial^2E^{(2,p)}}{\partial A^2}\right|_{A=0}
    \leq \dots \leq
    \left.\frac{\partial^2E_g}{\partial A^2}\right|_{A=0}
    = 4VD_\text{s}.
\end{equation}
For all examples we consider, this relation already holds starting at $p = 2$.

A schematic illustration of our argument is shown in Fig.~\ref{fig:Bootstrap_Superconductor_schematic}(b). A sequence of outer relaxations labeled by $(2,p)$ generally have different boundaries, but all boundaries meet at the FF point where the energy lower bounds are exact. As a result, the curvature of the $(2,p)$ boundaries lower bounds the true curvature of the set of physical 2RDMs (corresponding to the $(2,N)$ boundary).

In practice, we implement the following semidefinite program (SDP) using the solver \pkg{MOSEK}~\cite{mosek} together with \pkg{YALMIP}~\cite{YALMIP}:
\begin{equation}
\begin{split}
    \min_{\leftindex^2D}\quad &E(A)=\Tr{[H(A)\leftindex^2D]}\\
    \text{sub. to}\quad&
    \leftindex^2D, \leftindex^2G,\leftindex^2Q,\leftindex^3T_{1}, \leftindex^3T_{2}\succeq 0;\\
    &\Tr[\leftindex^2D]=\tfrac{N(N+1)}{2}.
\end{split}
\end{equation}
After obtaining the energy, we extract its curvature with respect to the gauge insertion $A$, yielding a lower bound on the stiffness.

\begin{figure}[t]
    \centering
    \includegraphics[width=0.95\linewidth]{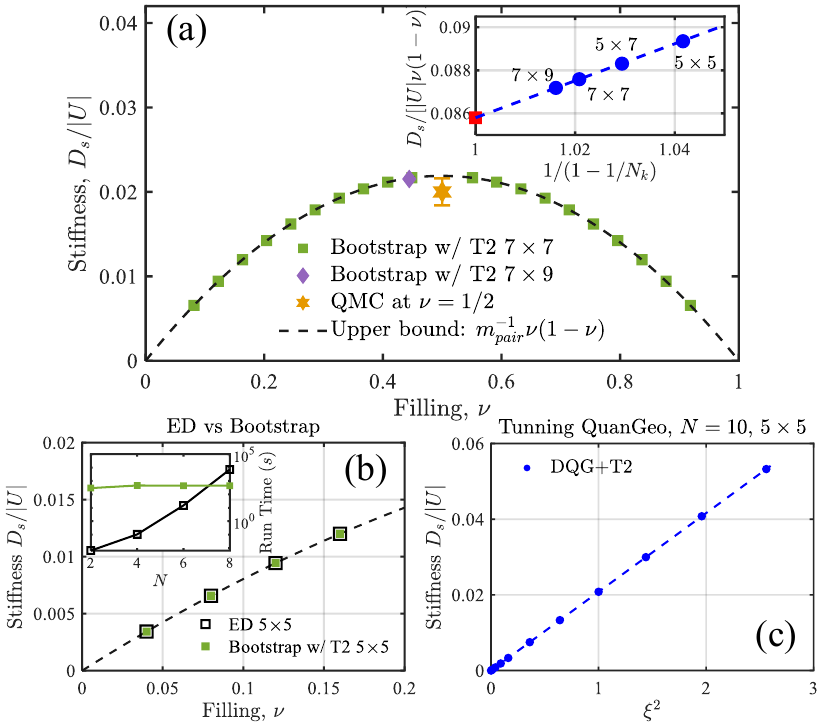}
    \caption{(a) Filling dependence of the superfluid stiffness in an attractive Hubbard model projected onto a topological FB (Model I~\cite{PhysRevB.102.201112}) for system size $7\times 7$ (squares). The square with a black outline denotes the exact result in the two-particle sector. The dashed curve shows $m_{\text{pair}}^{-1}\nu(1-\nu)$, where the pair mass $m_{\text{pair}}^{-1}$ is fixed by two-particle ED results for $7\times 7$. Determinant QMC data~\cite{Note_on_DQMC_data} are shown as a hexagram with error bar. The inset shows finite-size scaling of $D_\text{s}/[\nu(1-\nu)]$ obtained from bootstrapping different system sizes~\cite{Note_on_many_body_pair_mass}, indicating a correction of order $1/N_k$ to the many-body pair mass; the red dot marks the extrapolated thermodynamic value. (b) Comparison between exact diagonalization (ED) and bootstrap for the same system size ($5\times 5$) and fillings, showing perfect agreement. The inset shows, on a semi-log scale, the total elapsed time as a function of particle number $N$. (c) Quantum-geometry dependence of $D_\text{s}$ in an attractive Hubbard model projected onto a topologically trivial FB (Model II~\cite{PhysRevLett.130.226001}), where $\xi$ tunes the quantum geometry and the minimal quantum metric is $g=\xi^2/4$.} 
    \label{fig:filling_dependence}
\end{figure}

\textit{Models and Results.--}
We now present the results of the numerical bootstrap algorithm described above applied to a set of FF models for FB SCs belonging to the QGN family~\cite{PhysRevX.14.041004,Note_on_model_studied}. We consider two FB models: Model I describes a spinful FB with non-trivial spin Chern number~\cite{PhysRevB.102.201112}, and Model II describes a trivial spinful FB with tunable quantum geometry~\cite{PhysRevLett.130.226001}. Both models have spin $S^z$ conservation and time-reversal symmetry, which ensures the QGN condition for exact singlet-pairing SC GSs. We study the attractive Hubbard model~\eqref{FF_SC_Hamiltonian} in both cases, and the effects of additional magnetic FF interactions~\eqref{UJ model} for Model I. In all cases the interactions are projected onto the FB subspace, which is justified in the isolated-band limit where the gap to remote bands is much larger than the interaction strength.

Fig.~\ref{fig:filling_dependence} presents the results for the attractive Hubbard model, where panels (a,b) correspond to Model I and panel (c) to Model II. In Fig.~\ref{fig:filling_dependence}(a), we plot the filling dependence of the superfluid stiffness obtained from the bootstrap method (discrete data points) in comparison to the variational upper bound [Eq.~\eqref{lower_bound_formula}]. The pair mass $m_{\rm pair}$ is computed using exact diagonalization (ED) in the two-particle sector. The bootstrap calculations were performed for $7\times 7$ (green squares) and $7\times 9$ (purple diamonds). The finite-size scaling of the stiffness, shown in the inset of Fig.~\ref{fig:filling_dependence}(a), follows $\sim 1/(1-\frac{1}{N_k})$, allowing accurate extrapolation to the thermodynamic limit.

Surprisingly, we find that the bootstrap lower bound always saturates the variational upper bound within a small relative error of order $\sim 0.1\%$, suggesting that both methods give the \emph{exact} result. To verify this exactness in finite systems, we perform ED for a small $5\times 5$ system with few electrons and compare the results with the bootstrap, finding perfect agreement as shown in Fig.~\ref{fig:filling_dependence}(b). The inset shows the runtime of both ED and bootstrap as a function of particle number $N$, where ED scales exponentially with $N$, while the runtime of the bootstrap with T2 constraints is independent of $N$~\cite{Note_on_comparison_between_bootstrap_and_ED}. Furthermore, the extrapolated thermodynamic value for the ratio $D_\text{s}/[|U|\nu(1-\nu)]$ for Model I is $8.5799\times10^{-2}$ (red dot in the inset of Fig.~\ref{fig:filling_dependence}(a)), while the theoretical upper bound is $m_\text{pair}^{-1}/|U|=\frac{2 g}{N_\text{band}} = 8.5805\times10^{-2}$. As a reference, we also include the DQMC result extracted from Ref.~\cite{PhysRevB.102.201112} for the same system, which agrees with our result within error bars upon appropriate finite-size scaling~\cite{Note_on_DQMC_data}.

To further confirm this agreement, we apply our algorithm to Model II, which describes a trivial FB system with tunable minimal quantum metric $g=\xi^2/4$, where $\xi$ controls the spread of the Wannier orbitals~\cite{PhysRevLett.130.226001}. The results for $D_\text{s}$ at fixed filling $\nu$ and various $\xi$ are shown in Fig.~\ref{fig:filling_dependence}(c). We find that the lower bound on stiffness is strictly linear in $g$ and agrees \emph{quantitatively} with the variational upper bound~\eqref{lower_bound_formula}, thereby confirming the geometric origin of superfluid stiffness in this class of FB attractive Hubbard models.

\begin{figure}[t]
    \centering
    \includegraphics[width=0.85\linewidth]{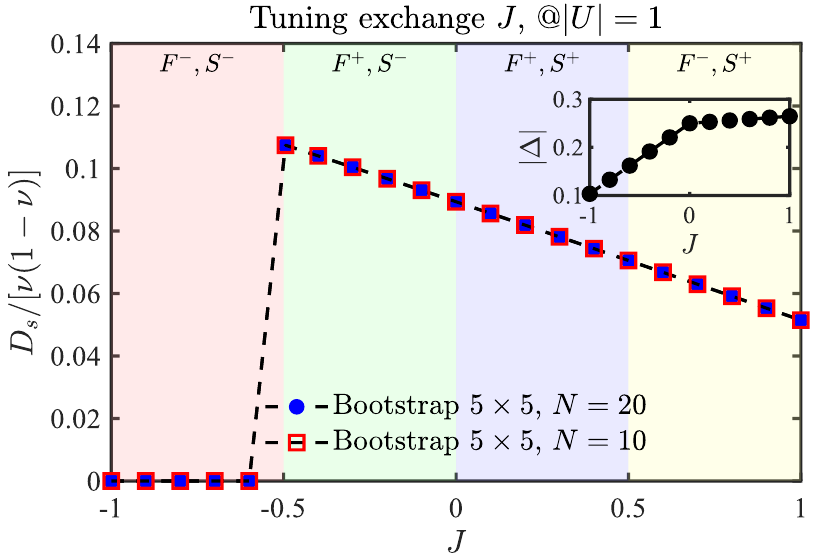}
    \caption{Tuning the additional nearest-neighbor $S^z$–$S^z$ interaction $J$ in a QGN model beyond the Hubbard limit (Model I$'$, see main text). The superfluid stiffness is shown as a function of $J$ with $|U|=1$ fixed. The plot is divided into four regions, $F^{+(-)}$ and $S^{+(-)}$, denoting (non-)FF and (non-)sign-problem-free regimes (e.g., $F^+,S^-$ indicates FF but sign problematic). Results are shown for two fillings, $N=20$ (blue) and $N=10$ (red), for a system size of $5\times 5$. The inset shows the exact single-particle excitation gap, which is independent of filling.}
    \label{fig:tuning_parameters}
\end{figure}

{\it Going beyond the Hubbard limit.}
So far, we have shown results for the attractive Hubbard model, which has been extensively studied and is amenable to QMC methods due to the absence of a sign problem~\cite{annurev:/content/journals/10.1146/annurev-conmatphys-033117-054307}. We now go beyond the Hubbard limit by introducing nearest-neighbor magnetic exchange ($S^z$-$S^z$) interactions with strength $J$, in addition to the Hubbard interactions in Model I [defined in Eq.~\eqref{UJ model}], which we denote as Model I$'$. Importantly, there is no obvious way to avoid the fermion sign problem for $J<0$ (ferromagnetic) interactions. However, the FF property of the model persists as long as $2|J|\leq |U|$, ensuring that our bootstrap algorithm remains applicable.

In Fig.~\ref{fig:tuning_parameters}, we show the superfluid stiffness $D_\text{s}$ of Model I$'$. We tune $J$ while fixing $|U|=1$ for two different fillings: $\nu=0.2$ (red) and $\nu=0.4$ (blue). We find that the calculated lower bound $D_\text{s}/[\nu(1-\nu)]$ again agrees with $(m_\text{pair})^{-1}$ and thus saturates the upper bound, corroborating the exactness of Eq.~\eqref{lower_bound_formula} (note that here $(m_\text{pair})^{-1} \neq 2g|U|/N_\text{band}$, unlike in the Hubbard case).
Moreover, the computed $D_\text{s}$ depends negatively on $J$, suggesting that additional ferromagnetic coupling enhances singlet SC order in topological FBs. In comparison, we show in the inset that the single-particle excitation gap (exactly calculable in QGN models~\cite{PhysRevX.14.041004}) exhibits the opposite trend but remains larger than the stiffness. This suggests a strong-coupling SC regime in which $T_c$ is controlled by the phase stiffness in the studied parameter range. These observations reveal an intricate interplay between interactions and FB quantum geometry in FB SC.

\textit{Discussions.--}
We note a caveat that, strictly speaking, the quantity we lower bound in our finite-size numerical simulations is the Drude weight, which contains contributions from both super- and normal-current responses. However, due to the presence of a charge (pairing) gap in the QGN models studied here, there is no normal current response at zero temperature. As a result, the Drude weight coincides with the superfluid stiffness~\cite{note_on_Drude_weight}.

We stress that our method is generally applicable to FF models beyond the QGN family. For example, in Ref.~\cite{han2025exactmodelschiralflatband} we use the same method to lower bound the stiffness in another FF model for FB SC within a single flavorless band. In that case, the simple relation between two-particle mass and stiffness in Eq.~\eqref{lower_bound_formula} no longer holds, and the lower bound does not saturate the upper bound.

Throughout this work, we implemented the DQG and T1/T2 constraints, but found that only the T2 constraint yields a non-trivial lower bound on the stiffness, which turns out to be exact. Since the RDM bootstrap bounds physical quantities using correlation functions, this suggests that the T2 constraint captures the essential correlations of the superconducting state in the presence of a phase twist.

Finally, we note that the present framework is easily generalized to lower bound susceptibilities in FF models. To do so, one can introduce a set of probing fields into any FF Hamiltonian,
\begin{align}
    \hat{H} = \hat{H}_\text{FF} + \sum_{i} \lambda_i \hat{h}_i ,
\end{align}
which perturbatively break the symmetries. One can then use the strategy introduced here to bound
\begin{equation}
\begin{split}
    &E_{\rm boot}|_{\boldsymbol\lambda=0}=  E_{g}|_{\boldsymbol\lambda=0}
    =\left.\frac{\partial E_{\rm boot}}{\partial \lambda_\mu}\right|_{\boldsymbol\lambda=0}
    = \left.\frac{\partial E_{g}}{\partial \lambda_\mu}\right|_{\boldsymbol\lambda=0} =0,\\
    &\Rightarrow \left[\chi_{\boldsymbol{\lambda}}\right]_{\mu\nu}\equiv
    \left.\frac{\partial^2 E_{\rm boot}(\boldsymbol\lambda)}{\partial \lambda_\mu\partial \lambda_\nu}\right|_{\boldsymbol\lambda=0}
    \preceq
    \left.\frac{\partial^2 E_{g}(\boldsymbol\lambda)}{\partial \lambda_\mu\partial \lambda_\nu}\right|_{\boldsymbol\lambda=0}.
\end{split}
\end{equation}
Here the first derivative vanishes due to symmetry, and $\chi_{\boldsymbol{\lambda}}$ is the negative of the physical susceptibility. This bound is complementary to variational bounds and may provide useful information on the location of symmetry-breaking phase transitions.

\begin{acknowledgments}
\textit{Acknowledgments.---} We thank Jonah Herzog-Arbeitman for very helpful discussions.
Z.~H. is supported by a Simons Investigator award, the
Simons Collaboration on Ultra-Quantum Matter, which is a grant from the Simons Foundation (Ashvin Vishwanath, 651440). E.~K. is supported by NSF MRSEC DMR-2308817 through the Center for Dynamics and Control of Materials. The authors thank the Harvard FAS Reaserch Computing (FASRC) for computational support. 

\textit{Data and code availability.---} Codes to reproduce the essential results are available at~\cite{Github_repo}.
\end{acknowledgments}

\bibliographystyle{plain}
\bibliographystyle{unsrt}

\begin{thebibliography}{10}
\bibitem{Note_on_sign_problem}
What we mean here is that the additional interactions breaks the symmetries that are used to show the model is sign-problem-free. Proving the model is no longer sign-problem-free is generally hard.
\bibitem{Note_on_DQMC_data}
The determinant QMC data is extracted from Ref.~\cite{PhysRevB.102.201112}. Some caveats need to be noted when comparing our results with the QMC result: 1) the QMC calculation was done without flat-band projection at $\nu=1/2$ and the $|U|$ used in their calculation is not small compared to the band gap, while we are studying a projected model neglecting the inter-band effects 2) although QMC method used in Ref.~\cite{PhysRevB.102.201112} is statistically `exact', the reported data has a relatively large numerical uncertainty; whereas our calculation is always almost in machine precision.
\bibitem{Note_on_many_body_pair_mass}
Due to some numerical reason, in this paper, we only implemented odd-by-odd system sizes. Thus, the data shown in the inset for different system sizes are obtained using different $\nu$'s since there is no common filling for them. They are $\nu_{5\times 5}=10/25/2$, $\nu_{5\times 7}=10/35/2$, $\nu_{7\times 7}=40/49/2$, and $\nu_{7\times 9}=56/63/2$. We note that the $\rho_s/(\nu(1-\nu))$ obtained from the bootstrap is filling $\nu$ independent up to numerical accuracy, so the results do not change when changing the fillings.
\bibitem{Note_on_model_studied}
It is worth noting that some FF models for superconductivity do not belong to this family e.g. Richardson-Gaudin models. However, all models considered in this work will have QGN.
\bibitem{Note_on_comparison_between_bootstrap_and_ED}
This is mainly because that the bootstrap we implemented stops at the certain level in the $(2,p)$ hierarchy, which scales roughly $N_k^{O(1)\cdot p}$ (in the T2 case $p=3$). We also provide the minimal example codes in this GitHub repository \cite{Github_repo} for the comparison between the bootstrap and ED.
\end{thebibliography}

\clearpage
\appendix
\onecolumngrid 
\vspace{4ex} 
{\centering\bfseries\large End Matter\par}
\vspace{4ex} 
\twocolumngrid 
\setcounter{secnumdepth}{1}

\section{A Theorem on the collapsing point of bootstrap feasible boundaries}\label{Theorem_bootstrap}
Let us denote any $N$-particle density matrix by $\hat \rho^N$. We denote the expectation value of any operator by
\begin{equation}
    \langle \hat O \rangle = { \Tr \hat \rho^N \hat O}
\end{equation}
For a $p$-body operator, the expectation value is obtained from $\langle \hat O \rangle = {\Tr \hat \rho^p_N \hat O}$ where $\hat\rho^p_N =\binom{N}{p}\Tr_{p+1\cdots N}[\hat\rho^N]$ is the $p$-reduced density operator ($p$-RDO). In particular, we will be interested in the 2-RDO defined as
\begin{equation}
    \hat\rho^2_N=\binom{N}{2}\Tr_{3\cdots N}[\hat\rho^N]\equiv \sum_{\alpha\beta\gamma\delta}\leftindex^2D^{\alpha\beta}_{\gamma\delta} |\alpha\beta\rangle\langle \gamma\delta|
\end{equation}
where $ |\alpha\beta\rangle$ is the 2-particle basis and $ \leftindex^2D$ is the 2RDM.
\begin{definition}[$(2,p)$ constraints]\label{2p_constraint}
    A constraint is called a $(2,p)$ constraint if:
    \begin{enumerate}
       \item it can be written as $\langle \sum_i\hat{O}_i\hat{O}^\dagger_i \rangle \geq 0$
        , where the maximal degree of $\hat{O}_i$'s is $p$, i.e., containing $p$ creation and/or annihilation operators.
        \item it is fully expressible in 2RDMs: $\langle \sum_i\hat{O}_i\hat{O}^\dagger_i \rangle := \Tr[\hat\rho^N\sum_i\hat{O}_i\hat{O}^\dagger_i] = \Tr[\hat\rho^2_N\sum_i\hat{O}_i\hat{O}^\dagger_i]$.
    \end{enumerate}
\end{definition}
\begin{remark}
    Condition 2 requires that the summation $\sum_i \hat{O}_i\hat{O}^\dagger_i$ cancels all terms beyond 2-body.
    If a $(2,p)$ constraint is saturated, then the corresponding 2RDM must lie on the boundary of $\leftindex^2_N{\mathbb{D}}^{(2,p)}$.
\end{remark}

\begin{definition}[$p$-body exactness/frustration-freeness]\label{p_body_exact}
A Hamiltonian is called $p_\mathbb{S}$-body exact Hamiltonian if up to a trivial constant shift:
\begin{enumerate}
  \item its ground state energy in $N$-particle sector is zero $\forall N\in\mathbb{S}$, where $ \mathbb{S}$ is a set of positive integers.
  \item it can be decomposed into sum of squares: $\hat H=\sum_i \hat{h}_i\hat{h}^\dagger_i$ where the maximal degree of $\hat{h}_i$'s is $p$.
\end{enumerate}
\end{definition}

\begin{remark}
  It is $p_N$-body exact if $\mathbb{S}=\{N\}$.
  A $p_{\mathbb{S}}$-body exact Hamiltonian can contain only $q$-body terms with $q\leq p$. Any Hamiltonian after a shift is trivially $N_N$-body exact: $\hat H=\sum_n(\varepsilon_n-\varepsilon_g) |\Psi^N_n\rangle\langle \Psi^N_n|$ with $|\Psi^N_n\rangle$ the $N$-particle eigenstates of energy $\varepsilon_n$. Thus, if a $p_N$-body exact Hamiltonian has $p<N$, it is non-trivial. 
\end{remark}

\begin{theorem}\label{boundary_collapse}
    An $N$-particle ground state 2RDM of a 2-body Hamiltonian with $p_N$-body exactness is a common boundary point of feasible regions $^2_N\mathbb{D}^{(2,q)}$ $\forall q: p\leq q\leq N$. 
\end{theorem}
\begin{proof}
    By definition~\ref{p_body_exact}, any $N$-particle ground state, $|\Psi^N_g\rangle$ (with $\hat\rho^N_g=|\Psi^N_g\rangle\langle\Psi^N_g|$) of any two-body Hamiltonian $\hat H$ with $p_N$-body exactness satisfies the null condition:
    \begin{equation}
       \Tr[\hat\rho^N_g\hat H] =\langle \Psi^N_g |H|\Psi^N_g\rangle =  \langle \Psi^N_g |\sum_i \hat{h}_i\hat{h}^\dagger_i|\Psi^N_g\rangle = 0,
    \end{equation}
    where $\sum_i \hat{h}_i\hat{h}^\dagger_i$ contains only two-particle terms with degree $p$ operators $\hat{h}_i$. This implies $\hat h^\dagger_i|\Psi^N_g\rangle=0$ $\forall i$.
    Thus, if we consider the following constraint: $\Tr[\hat\rho^2_N\sum_i\hat{h}_i\hat{h}^\dagger_i]\geq 0$,
    which by definition~\ref{2p_constraint} is a $(2,p)$ constraint, it is saturated by the ground state 2-particle reduced density operator: $\hat\rho^2_g=\binom{N}{2}\Tr_{3\cdots N}[\hat\rho^N_g]$ meaning that its corresponding 2RDM $\leftindex^2D_g$ is on the ($2,p$) boundary.
    Now, we lift the constraint by inserting the number operator $q-p$ times~\cite{PhysRevLett.108.263002}
    \begin{equation}
        \langle \Psi^N |\sum_{ij_1,\cdots,j_{q-p}} \left[\hat{h}_ic^\dagger_{j1}\cdots c^\dagger_{j_{q-p}}\right]\left[\hat{h}_ic^\dagger_{j1}\cdots c^\dagger_{j_{q-p}}\right]^\dagger|\Psi^N\rangle\geq 0,
    \end{equation}
    which is a $(2,q)$ constraint since it involves operators of degree $q$ and still expressible in 2RDMs. It is again saturated by the ground state 2RDM $\leftindex^2D_g$ since $\hat h^\dagger_i|\Psi^N_g\rangle=0$. Thus, $\leftindex^2D_g$ must be the common boundary point of all feasible regions at levels $(2,q)$ $\forall q: p\leq q\leq N$.
\end{proof}

\begin{remark}
    The FF models discussed in this work is $2_{2\mathbb{Z}^+}$-body exact. Thus, in any fixed even particle number sector (even $N$), we have that all the $(2,p)$ boundaries for $p=2,\cdots,N$ collaps to a single point, the 2RDM of the ground state which we refer as the FF point (see the red dot in Fig.~\ref{fig:Bootstrap_Superconductor_schematic}(b)). This also means that the lower bound on the ground state energy for those FF models is exact when imposing the $(2,p)$ constaints with $p\geq2$.
\end{remark}

\begin{figure}
    \centering
    \includegraphics[width=0.9\linewidth]{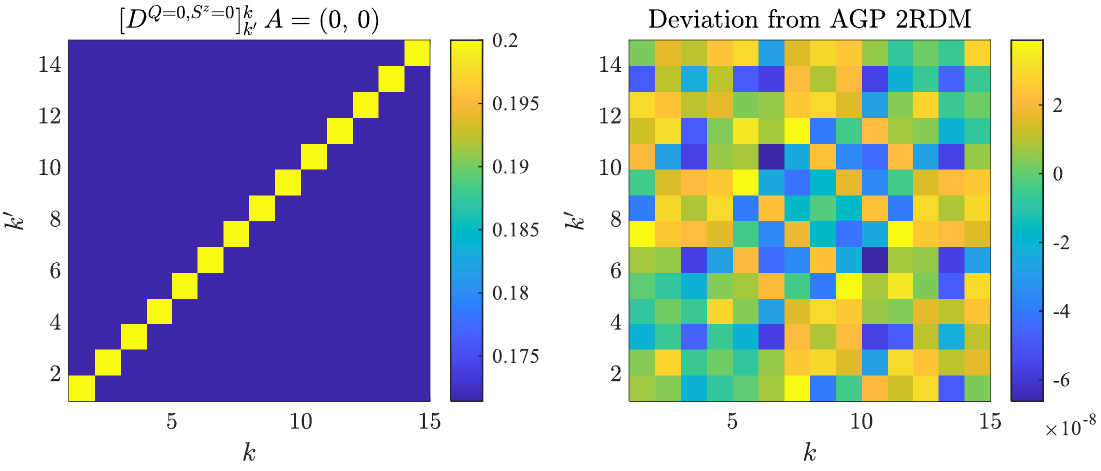}
\caption{The 2RDM in $\bQ=0,S^z=0$ sector at the FF point obtained from bootstrap with T2 (left panel) and its deviation from the exact AGP 2RDM~\eqref{AGP_2RDM} (right panel), indicating a perfect agreement. }
    \label{fig:2RDM}
\end{figure}

\section{Details about the model and the ground state}\label{FF_SC}

In this appendix, we present the details about the models used in the main text. There are two band structures used in this work. The first (referred to as Model I  in the main text when studied with Hubbard interaction) is a topological FB given by the lower band of the hopping matrix~\cite{PhysRevB.102.201112}:
\begin{equation}\label{Model_I}
    H_{\rm I}(\bk) = B_\bk^0\sigma^0 +\bB_\bk\cdot\bsigma
\end{equation}
with
\begin{equation}
\begin{split}
B_{\bk}^{x} + i\,B_{\bk}^{y} 
  =& -t_{1}\left[\,e^{i\pi/4} + \,e^{\mathrm{i}(k_x+k_y) +i\pi/4}\right.\nonumber\\
  & \ \ \ \ \ \ \ \ \ \ \ \left. + \,e^{\mathrm{i}k_x -i\pi/4}\,+\,e^{\mathrm{i}k_y -i\pi/4} \right],\\
B_{\bk}^{z} 
  =& -2t_{2}\,\bigl[\cos(k_{y})-\cos(k_{x})\bigr],\\
B_{\bk}^{0} 
  =& -2t_{5}\,\bigl[\cos\!\bigl(2(k_{x})\bigr)
      +\cos\!\bigl(2(k_{y})\bigr)\bigr],
\end{split}
\end{equation}
where we have chosen the periodic embedding, and we approach the flat-band limit by taking: $t_1 = 1$, $t_2 = 1/\sqrt{2}$, and $t_5 = (1-\sqrt{2})/4$. The other FB (Model II) is topologically trivial but has a tunable quantum-geometry~\cite{PhysRevLett.130.226001}:
\begin{equation}
\begin{split}\label{Model_II}
H_{\rm II}(\bk)&= 
  -t 
  \bigl(\,
      \sigma^{x}\,\sin\alpha_{\bk}+
      \sigma^{y}\,\cos\alpha_{\bk}+
      \mu\sigma^{0}
  \bigr)\\
\alpha_{\bk} &= \xi\bigl(\cos k_{x} + \cos k_{y}\bigr),
\end{split}
\end{equation}
where we take $t=1$ and $\mu=0$. 

For both cases, the lower band is a spinful, time-reversl symmetric, $S^z$ conserving flat-band, and thus satisfy quantum geometric nesting (QGN)~\cite{PhysRevX.14.041004} in the particle-particle channel at zero momentum. This nesting structure allows for an infinite set of solvable model construction. The simplest one is the attractive Hubbard model:
\begin{equation}\label{FF_SC_Hamiltonian}
    \begin{split}
    \hat{H} &= \frac{|U|}{2}  \sum_{\vec{{R}} \alpha }\left({\bar{n}}_{ \vec{R} \alpha \uparrow}  - {\bar{n}}_{ \vec{R} \alpha \downarrow}\right)^2,
    \end{split}
\end{equation}
where ${\bar{n}}_{ \vec{R} \alpha \sigma}$ is the number operator of orbital $\alpha$ at unit cell at $\vec{R}$ with spin $\sigma$, projected onto the FB subspace. 

The second QGN interacting model we study for the band structure in Model I (referred to as Model I'), is
\begin{align}\label{UJ model}
    \hat{H} &= \frac{|U|}{2}  \sum_{\vec{{R}} \alpha }\left({\bar{n}}_{ \vec{R} \alpha \uparrow}  - {\bar{n}}_{ \vec{R} \alpha \downarrow}\right)^2+  J\sum_{\langle ij\rangle} \hat{S}^z_{i} \hat{S}^z_j ,
\end{align}
which includes an additional magnetic coupling $J$ that is among nearest neighbor orbitals. To ensure FF-ness, we demand $|J| \leq \frac{|U|}{2}$.

In these simple cases, the projection onto the FB subspace can be implemented by keeping terms that only involve electron operators within the FBs (note that this is generically not true, see Ref.~\cite{PhysRevX.14.041004} for detailed discussions):
\begin{align}
    \hat{\gamma}_{\vec{k}n\in \mathcal{F} \sigma} &= U_{n\alpha}^{\sigma,\dagger}(\vec{k}) \hat{c}_{\vec{k}\alpha\sigma},
\end{align}
where $\mathcal{F}$ is the set of FB indices, and $U_{\alpha n}^\sigma(\vec{k})$ is the $n$-th band wavefunction. The coupling to flat gauge connection is achieved by replacing $\vec{k} \rightarrow \vec{k}+\vec{A}$ in the above $U$ argument. For completeness, we include the expression for the used projected Hamiltonian in momentum space:
\begin{equation}
    \begin{split}
    \hat{H}=&
 \frac{1}{8\mathsf{V}}\sum_{\substack{\vec{k}_1+\vec{k}_3 = \vec{k}_2 + \vec{k}_4, \pm, \alpha\beta \\ nmkl \in \mathcal{F}, \sigma \sigma',I =1\dots 4}}  \left(\frac{|U|}{2} \pm J\right)  (-1)^{\sigma +\sigma'} \\
 & \times D^{(I,\pm)}_{\alpha\alpha} (\vec{k}_1-\vec{k}_2)D^{(I,\pm)}_{\beta\beta}(\vec{k}_3-\vec{k}_4) \\
&\times\hat{\gamma}^\dagger_{\vec{k}_1 n \sigma }U_{n \alpha}^{\sigma,\dagger}(\vec{k}_1+\vec{A}) U_{\alpha m}^{\sigma}(\vec{k}_2+\vec{A}) \hat{\gamma}_{\vec{k}_2 m\sigma }\\
&\times\hat{\gamma}^\dagger_{\vec{k}_3 k\sigma' } U_{k\beta}^{\sigma' ,\dagger}(\vec{k}_3+\vec{A}) U_{\beta l}^{\sigma' }(\vec{k}_4+\vec{A})\hat{\gamma}_{\vec{k}_4 l\sigma'}
\end{split}
\end{equation}

where
\begin{align}
    D^{(I,\pm)}_{AA} (\vec{q}) =& 1 \\
    D^{(I,\pm)}_{BB} (\vec{q}) =&  \pm \begin{cases}
    1 & I=1\\
         e^{- \mathrm{i}\vec{q}_x} & I=2 \\
         e^{- \mathrm{i}\vec{q}_y}  & I=3 \\
         e^{- \mathrm{i}(\vec{q}_x+\vec{q}_y)}  & I=4
    \end{cases}
\end{align}

It can be shown that the Hamiltonian commutes with a uniform $s$-wave pair creation operator~\cite{PhysRevX.14.041004}:
\begin{align}
    \eta \equiv \sum_{\vec{k}} c_{\bk\uparrow}c_{-\bk\downarrow}
\end{align}
when there is no external gauge field $\bA= 0$. Therefore, this is a FF model whose ground state in any even particle number $N\in 2\mathbb{Z}$ sector has zero energy $E=0$ and takes a special form called anti-symmetrized geminal power (AGP)~\cite{coleman1965structure}:
\begin{equation}
    |\text{AGP}\rangle = \frac{1}{(N/2)!}\left( \eta^\dagger \right)^{N/2}|\text{Vac}\rangle.
\end{equation}
This then allows us to analytically write down the exact 2RDM of such ground state~\cite{khamoshi2019efficient}, specifically, in the $\bQ$ and total $S^z=0$-sector:
\begin{equation}\label{AGP_2RDM}
    \begin{split}
    &\left[\leftindex^2D_{\text{AGP}}^{\bQ,S^z=0}\right]^\bk_{\bk'} = \langle\text{AGP}|c^\dagger_{\bk,\uparrow} c^\dagger_{\bQ-\bk,\downarrow}c_{\bQ-\bk',\downarrow}c_{\bk',\uparrow}|\text{AGP}\rangle \\
    & = \delta_{\bQ,0}[C_1\delta_{\bk,\bk'}+C_3(1-\delta_{\bk,\bk'})]+(1-\delta_{\bQ,0})C_2\delta_{\bk,\bk'}
    \end{split}
\end{equation}
with $C_1=N_{pair}/N_o,C_2 =N_{pair}(N_{pair}-1)/N_o(N_o-1)$, and $C_3 =N_{pair}(N_o-N_{pair})/N_o(N_o-1)$, which can be used as a check for the bootstrap at $\bA =0$ (see Fig.~\ref{fig:2RDM}).

\clearpage          
\onecolumngrid      



\section*{Supplementary material}

# Supplemental Materials for "Bootstrapping Flat-band Superconductors: Rigorous Lower Bounds on Superfluid Stiffness"

Qiang Gao, Zhaoyu Han, and Eslam Khalaf

*Department of Physics, Harvard University, Cambridge, Massachusetts 02138, USA*

In this note, we derive a result quoted in the main text: the superfluid stiffness is upper bounded by

$$\rho_{\rm s} \leq \frac{{\sf N}_{\rm flat}}{2}(m_{\rm pair})^{-1}\nu(1-\nu) \tag{1}$$

in quantum geometric nesting models [1].

## CONTENTS

## I. REVIEW ON THE QUANTUM GEOMETRIC NESTING MODELS

### A. Preliminaries

We consider the electronic structure encoded in a tight-binding model with translation symmetry:

$$\hat{H}_0 = \sum_{\bm{R}\bm{R}',\alpha\beta} t_{\alpha\beta,\bm{R}-\bm{R}'} \hat{c}^\dagger_{\bm{R},\alpha} \hat{c}_{\bm{R}',\beta} \tag{2}$$

where $\bm{R}$ indicate the unit cell center position, Greek letters $\alpha, \beta, \ldots$ denote the orbital index in each unit cell. For simplicity, we combine the orbital and spin indices unless otherwise specified. In momentum space, we define basis $\hat{c}_{\bm{k},\alpha} \equiv \sum_{\bm{R}} e^{i\bm{k}\cdot(\bm{R}+\bm{r}_\alpha)} \hat{c}_{\bm{R},\alpha}/\sqrt{\sf V}$ and Fourier transform $t_{\alpha\beta}(\bm{k}) \equiv \sum_{\bm{R}} e^{i\bm{k}\cdot(\bm{R}+\bm{r}_\alpha-\bm{r}_\beta)} t_{\alpha\beta,\bm{R}}$ (${\sf V}$ is the number of unit cells in the system, $\bm{r}_\alpha$ is the intra-unit-cell position of orbital $\alpha$) to rewrite

$$\hat{H}_0 = \sum_{\bm{k},\alpha\beta} t_{\alpha\beta}(\bm{k}) \hat{c}^\dagger_{\bm{k},\alpha} \hat{c}_{\bm{k},\beta} \tag{3}$$

Then we can diagonalize the Hamiltonian into the form

$$\hat{H}_0 = \sum_{\bm{k},n} \epsilon_n(\bm{k}) \hat{\gamma}^\dagger_{\bm{k},n} \hat{\gamma}_{\bm{k},n}, \tag{4}$$

$$\hat{\gamma}_{\bm{k},n} = U^\dagger_{n\alpha}(\bm{k}) \hat{c}_{\bm{k},\alpha}, \tag{5}$$

where we use Latin letters $(n, m, \ldots)$ to label the bands. We assume there are ${\sf N}$ orbitals, and thus ${\sf N}$ bands in total.

We consider an ideal setup where there are $\mathsf{N}_{\text{flat}}$ nearly degenerate and flat bands isolated from all the other bands by a large gap $\Delta$, and the relevant interaction strength $V$ is small compared to $\Delta$ but large compared to the flat bands' bare bandwidth $W$:

$$\Delta \gg V \gg W \tag{6}$$

Without loss of generality, we can relabel all the bands and assume the $n = 1 \ldots \mathsf{N}_{\text{flat}}$ bands are the degenerate flat bands, and $n = \mathsf{N}_{\text{flat}} + 1 \ldots \mathsf{N}$ bands are the others.

For the single-particle basis, we can thus define a projector matrix:

$$P_{\alpha\beta}(\boldsymbol{k}) \equiv \sum_{n \leq \mathsf{N}_{\text{flat}}} U_{\alpha n}(\boldsymbol{k}) U^\dagger_{n\beta}(\boldsymbol{k}), \tag{7}$$

Similarly, we also define a projector matrix onto the remaining bands, $Q_{\alpha\beta}(\boldsymbol{k}) \equiv \delta_{\alpha\beta} - P_{\alpha\beta}(\boldsymbol{k})$.

For the ideal scenario we are assuming, we define the "flat band subspace", $\mathcal{H}_{\text{flat}}$, as the Fock space of the electron modes $\gamma_{\boldsymbol{k}, n=1\ldots\mathsf{N}_{\text{flat}}}$ tensoring with the "vacuum" state of the other modes $\gamma_{\boldsymbol{k}, n=\mathsf{N}_{\text{flat}}+1\ldots\mathsf{N}}$ (i.e. all the modes with energy higher/lower than the flat bands are empty/occupied). Especially we will define $|\text{vac}\rangle$ of $\mathcal{H}_{\text{flat}}$ as the "vacuum" state for the remote bands tensoring the empty state for the flat bands.

We formally define the projector onto $\mathcal{H}_{\text{flat}}$ as $\hat{P}$ and $\hat{Q} \equiv 1 - \hat{P}$. To the leading order of a perturbation series treating both $V/\Delta$ and $W/\Delta$ as small parameters, we can set $\epsilon_{n=1\ldots\mathsf{N}_{\text{flat}}}(\boldsymbol{k}) = 0$ and project the interacting Hamiltonian onto $\mathcal{H}_{\text{flat}}$. The solvable model we construct should be viewed as exact up to $\mathcal{O}[(V/\Delta)^2, (W/V)]$ errors.

For later convenience, we define a $(\mathsf{N} - \mathsf{N}_{\text{flat}}) \times (\mathsf{N} - \mathsf{N}_{\text{flat}})$ diagonal matrix $\Lambda$ denoting the occupation of the remote bands ($\boldsymbol{k}$ arbitrary in the below definition):

$$\Lambda = \text{diag}\left(\langle \gamma^\dagger_{\boldsymbol{k},\mathsf{N}_{\text{flat}}+1} \gamma_{\boldsymbol{k},\mathsf{N}_{\text{flat}}+1}\rangle, \langle \gamma^\dagger_{\boldsymbol{k},\mathsf{N}_{\text{flat}}+2} \gamma_{\boldsymbol{k},\mathsf{N}_{\text{flat}}+2}\rangle, \ldots \right) \tag{8}$$

We emphasize that most results of discussed below do not rely on symmetries other than translation symmetry and charge conservation, and they apply to flat-band systems in arbitrary dimensions.

## B. The QGN models

Ref. [1] constructs a class of exactly solvable model, call the quantum geometric nesting (QGN) models, for such isolated flat band systems whenever a nesting condition is satisfied. Focusing on the superconducting case, the ground states in the charge-$2N$ sector is of the form:

$$|N\rangle = (\eta^\dagger)^N |\text{vac}\rangle. \tag{9}$$

which is defined by a pairing operator

$$\hat{\eta}^\dagger \equiv \frac{1}{\mathsf{V}} \sum_{\boldsymbol{k}; n,m \leq \mathsf{N}_{\text{flat}}} F^{\boldsymbol{Q}}_{nm}(\boldsymbol{k}) \hat{\gamma}^\dagger_{\boldsymbol{Q}/2+\boldsymbol{k},n} \hat{\gamma}^\dagger_{\boldsymbol{Q}/2-\boldsymbol{k},m} \tag{10}$$

where $\boldsymbol{Q}$ is a momentum determined by a satisfied nesting condition and $F^{\boldsymbol{Q}}_{nm}(\boldsymbol{k})$ is the form factor defined in the flat band subspace satisfying

$$F^{\boldsymbol{Q}}(\boldsymbol{k}) = -F^{\boldsymbol{Q},T}(-\boldsymbol{k}) \tag{11}$$

For simplicity we will set $\boldsymbol{Q} = 0$ below.

For each possible pairing operator, there is an infinite set of interacting Hamiltonians for which Eq. 9 are the ground states. In general, those "interactions" take the form

$$\hat{H}_{\text{int}} \equiv \sum_{\boldsymbol{R},\boldsymbol{R}'} V_{\boldsymbol{R}-\boldsymbol{R}'} \hat{S}_{\boldsymbol{R}} \hat{S}_{\boldsymbol{R}'} + \hat{H}_{\text{int},2} \tag{12}$$

where $\hat{S}_{\boldsymbol{R}}$ a hermitian, local, neutral, fermion bilinear operator (of the form $\hat{c}^\dagger \hat{c}$) centered at $\boldsymbol{R}$. For the hermicity of the Hamiltonian, we need $V_{\boldsymbol{R}-\boldsymbol{R}'} = [V_{\boldsymbol{R}'-\boldsymbol{R}}]^*$ and thus $V(\boldsymbol{q})$ is real on all $\boldsymbol{q}$; besides this, we further require $V(\boldsymbol{q})$

to be *non-negative* on all $q$. $\hat{H}_{\text{int},2}$ contain some suitably designed fermion bilinear terms such that the projected Hamiltonian takes the form

$$\hat{P}\hat{H}\hat{P} = \sum_{\boldsymbol{R},\boldsymbol{R}'} V_{\boldsymbol{R}-\boldsymbol{R}'} \hat{P}\hat{S}_{\boldsymbol{R}}\hat{P}\hat{S}_{\boldsymbol{R}'}\hat{P} \tag{13}$$

The crucial condition that underlies the solvability of the model is that $\hat{S}_{\boldsymbol{R}}$ needs to commute with the pair creation operator within the $\mathcal{H}_{\text{flat}}$:

$$[\hat{\eta}^\dagger, \hat{P}\hat{S}_{\boldsymbol{R}}\hat{P}] = 0 \tag{14}$$

It should be noted that these operators are defined in the entire Hilbert space, whereas the above commutation relation only holds after projection. Different $\hat{S}_{\boldsymbol{R}}$ at different $\boldsymbol{R}$ are not necessarily mutually commuting. There are infinitely many such constructions for each pair operator, so that any combination of the corresponding models is also a valid solvable model.

For Eq. 14 to be true, there is a condition that needs to be satisfied by $\hat{S}_{\boldsymbol{R}}$. To see this, we write

$$\hat{S}_{\boldsymbol{R}} = \sum_{\mu\nu,\boldsymbol{R}_1,\boldsymbol{R}_2} S_{\mu\nu}(\boldsymbol{R}_1,\boldsymbol{R}_2)\hat{c}^\dagger_{\boldsymbol{R}+\boldsymbol{R}_1,\mu}\hat{c}_{\boldsymbol{R}+\boldsymbol{R}_2,\nu} \tag{15}$$

which can be written in band basis as

$$\hat{S}_{\boldsymbol{R}} = \frac{1}{\mathsf{V}} \sum_{\mu\nu,\boldsymbol{pq}} e^{i\boldsymbol{R}\cdot(\boldsymbol{p}-\boldsymbol{q})} S_{\mu\nu}(\boldsymbol{p},\boldsymbol{q})\hat{c}^\dagger_{\boldsymbol{p},\mu}\hat{c}_{\boldsymbol{q},\nu} \tag{16}$$

$$= \frac{1}{\mathsf{V}} \sum_{nm,\boldsymbol{pq}} e^{i\boldsymbol{R}\cdot(\boldsymbol{p}-\boldsymbol{q})} S_{nm}(\boldsymbol{p},\boldsymbol{q})\hat{\gamma}^\dagger_{\boldsymbol{p},n}\hat{\gamma}_{\boldsymbol{q},m} \tag{17}$$

where

$$S_{\mu\nu}(\boldsymbol{p},\boldsymbol{q}) \equiv \sum_{\boldsymbol{R}_1,\boldsymbol{R}_2} S_{\mu\nu}(\boldsymbol{R}_1,\boldsymbol{R}_2) e^{i[(\boldsymbol{R}_1+\boldsymbol{r}_\mu)\cdot\boldsymbol{p}-(\boldsymbol{R}_2+\boldsymbol{r}_\nu)\cdot\boldsymbol{q}]} \tag{18}$$

$$S_{nm}(\boldsymbol{p},\boldsymbol{q}) \equiv U^\dagger_{n\mu}(\boldsymbol{p}) S_{\mu\nu}(\boldsymbol{p},\boldsymbol{q}) U_{\nu m}(\boldsymbol{q}). \tag{19}$$

The hermicity of $\hat{S}_{\boldsymbol{R}}$ implies that

$$S(\boldsymbol{p},\boldsymbol{q}) = S^\dagger(\boldsymbol{q},\boldsymbol{p}) \tag{20}$$

Then, the condition that needs to be satisfied for Eq. 14 is:

$$\sum_{m \leq \mathsf{N}_{\text{flat}}} [-F_{km}(\boldsymbol{p}) S_{nm}(-\boldsymbol{q},-\boldsymbol{p}) - S_{km}(\boldsymbol{p},\boldsymbol{q}) F_{mn}(\boldsymbol{q})] = 0 \quad \forall \boldsymbol{p},\boldsymbol{q} \tag{21}$$

Or more compactly,

$$F^T(-\boldsymbol{p}) S^T(-\boldsymbol{q},-\boldsymbol{p}) = S(\boldsymbol{p},\boldsymbol{q}) F(\boldsymbol{q}) \quad \forall \boldsymbol{p},\boldsymbol{q} \tag{22}$$

For completeness here we give explicit expression for the ideal solvable Hamiltonian in momentum space. We first recast the terms in momentum space:

$$\hat{H}_{\text{int}} = \mathsf{V} \sum_{\boldsymbol{K}} V(\boldsymbol{K}) \hat{S}^\dagger(\boldsymbol{K}) \hat{S}(\boldsymbol{K}) \tag{23}$$

where

$$V(\boldsymbol{K}) \equiv \sum_{\boldsymbol{R}} e^{-i\boldsymbol{K}\cdot\boldsymbol{R}} V_{\boldsymbol{R}} \tag{24}$$

$$\hat{S}(\boldsymbol{K}) \equiv \frac{1}{\mathsf{V}} \sum_{\boldsymbol{R}} e^{-i\boldsymbol{K}\cdot\boldsymbol{R}} \hat{S}_{\boldsymbol{R}} \tag{25}$$

$$= \frac{1}{\mathsf{V}} \sum_{nm,\boldsymbol{pq}} \delta_{\boldsymbol{K},\boldsymbol{p}-\boldsymbol{q}} S_{nm}(\boldsymbol{p},\boldsymbol{q}) \hat{\gamma}^\dagger_{\boldsymbol{p},n} \hat{\gamma}_{\boldsymbol{q},m} \tag{26}$$

thus the projected Hamiltonian reads:

$$\hat{H}_{\text{solvable}} = \frac{1}{V} \sum_{\substack{\boldsymbol{pq}, \boldsymbol{p'q'} \\ nmkl \leq N_{\text{flat}}}} \delta_{\boldsymbol{p'}-\boldsymbol{q'}, \boldsymbol{q}-\boldsymbol{p}} V(\boldsymbol{q}-\boldsymbol{p}) S^*_{mn}(\boldsymbol{q}, \boldsymbol{p}) S_{kl}(\boldsymbol{p'}, \boldsymbol{q'}) \hat{\gamma}^\dagger_{\boldsymbol{p},n} \hat{\gamma}_{\boldsymbol{q},m} \hat{\gamma}^\dagger_{\boldsymbol{p'},k} \hat{\gamma}_{\boldsymbol{q'},l} \quad (27)$$

$$= \frac{1}{V} \sum_{\substack{\boldsymbol{pq}, \boldsymbol{p'q'} \\ nmkl \leq N_{\text{flat}}}} V_{nmkl}(\boldsymbol{p}, \boldsymbol{q}, \boldsymbol{p'}, \boldsymbol{q'}) \hat{\gamma}^\dagger_{\boldsymbol{p},n} \hat{\gamma}_{\boldsymbol{q},m} \hat{\gamma}^\dagger_{\boldsymbol{p'},k} \hat{\gamma}_{\boldsymbol{q'},l} \quad (28)$$

## II. RIGOROUS UPPER BOUND ON STIFFNESS

### A. The strategy

The superfluid stiffness is defined by the energy response to a flat gauge connection $\boldsymbol{A}$:

$$\kappa_{ij} = \lim_{V \to \infty} \frac{1}{4V} \left. \frac{\partial^2 E(\boldsymbol{A})}{\partial A_i \partial A_j} \right|_{\boldsymbol{A}=0} \quad (29)$$

where $E(\boldsymbol{A})$ is the ground state energy in the presence of constant vector potential $\boldsymbol{A}$. It is important to note that the order of taking the thermodynamic limit and the derivatives matters, since when a component $A_i$ is a multiple of $2\pi/L_i$, its effect can be completely absorbed by a rigid shift of momentum grid ($L_i$ is the linear system size in that direction). For a finite system size and a small but nonzero $|\boldsymbol{A}| \ll \frac{2\pi}{L}$, we cannot exactly solve the model but can use variational method to upper bound $E(\boldsymbol{A})$. Specifically, any variational state yields an upper bound

$$E_{\text{var}}(\boldsymbol{A}) = \min_F \frac{\langle N|\hat{H}(\boldsymbol{A})|N\rangle_{\text{var}}}{\langle N|N\rangle_{\text{var}}} \geq E(\boldsymbol{A}) \quad (30)$$

In this work we will use BCS states or projected BCS states (also known as Antisymmetric Geminal Power - AGP -states) as our variational ansatz. When $\boldsymbol{A} = 0$, we know our model is solvable and the true GS is indeed a BCS/AGP state, so the variational ansatz is exact at $\boldsymbol{A} = 0$, where the energy is simply 0. On the other hand, by gauge invariance, we know that $\partial_{A_i} E = 0$ at $\boldsymbol{A} = 0$ for both the true energy and the variational energy, so that the quadratic order in the Taylor's expansion is the leading non-zero order. Combining these two facts, we reach the conclusion:

$$\kappa_{\text{var}} \succeq \kappa \succeq 0 \quad (31)$$

meaning that the stiffness tensor obtained by variational energy is greater or equal to the real one, thus providing a rigorous upper bound ($H(\boldsymbol{A})$ is always positive, so $\kappa$ is guaranteed to be non-negative). In particular, this implies the geometric mean of the stiffness $\bar{\kappa} \equiv \sqrt[d]{\det \kappa}$, which is a relevant quantity for the coherence energy scale of superconductor, is rigorously bounded from above

$$\sqrt[d]{\det \kappa_{\text{var}}} \geq \sqrt[d]{\det \kappa}. \quad (32)$$

### B. The variational ansatzes and the observables

Here we discuss the variational ansatz used in the below calculations. They are projected BCS (or antisymmetric geminal power) states

$$|N\rangle = \left(\hat{\eta}^\dagger\right)^N |\text{vac}\rangle \quad (33)$$

$$\hat{\eta}^\dagger \equiv \sum_{\boldsymbol{k}; n, m \leq N_{\text{flat}}} F_{nm}(\boldsymbol{k}) \hat{\gamma}^\dagger_{\boldsymbol{k},n} \hat{\gamma}^\dagger_{-\boldsymbol{k},m} \quad (34)$$

and BCS states defined by an additional complex number $\alpha$:

$$|\alpha\rangle \equiv \sum_N \frac{e^{-\alpha N}}{2} |N\rangle \quad (35)$$

where $F$ is variational form factor satisfying $F(\mathbf{k}) = -F^T(-\mathbf{k})$.

To make the structure of these states clear, we singular-value decompose the form factor $F(\mathbf{k})$ as

$$F_{nm}(\mathbf{k}) = \sum_i F_{nm}^{(i)}(\mathbf{k}) = \sum_i f^{(i)}(\mathbf{k}) v_n^{(i)}(\mathbf{k}) v_m^{(i)}(-\mathbf{k}) \tag{36}$$

where $f^{(i)}(\mathbf{k}) = -f^{(i)}(-\mathbf{k})$ and on each $\mathbf{k}$, the set of vectors $\{v^{(i)}\}$ are orthonormal. Defining a new basis for the flat band orbitals

$$\hat{\xi}_{\mathbf{k},i}^\dagger = v_n^{(i)}(\mathbf{k}) \hat{\gamma}_{\mathbf{k},n}^\dagger \tag{37}$$

we can recast the $\hat{\eta}$ operator into a compact form:

$$\hat{\eta}_{\text{var}}^\dagger \equiv 2 \sum_{\mathbf{k}; i=1\ldots N_{\text{flat}}}' f^{(i)}(\mathbf{k}) \hat{\xi}_{\mathbf{k},i}^\dagger \hat{\xi}_{-\mathbf{k},i}^\dagger \tag{38}$$

The $\sum'$ is not over the whole momentum space but only half of it due to the redundancy of the representation; the prefactor 2 also arise for the same reason. Then each pair $\hat{\xi}_{\mathbf{k},i}^\dagger \hat{\xi}_{-\mathbf{k},i}^\dagger$ can be viewed as the creation operator onto a hard-core boson mode [2]. With this simple interpretation, we find the norm of the state

$$\mathcal{N} \equiv \langle N|N\rangle_{\text{var}} = 4^N S_N^{\mathsf{N}_0}(\{|f^{(i)}(\mathbf{k})|^2\}) \tag{39}$$

where $S_N^{\mathsf{N}_0}$ is the Elementary Symmetric Polynomial (ESP) of degree $N$ with $\mathsf{N}_0$ variables $\{|f^{(i)}(\mathbf{k})|^2\}$ ($\mathbf{k}$ is again restricted in half of the momentum space), and $\mathsf{N}_0 \equiv \mathsf{N}_{\text{flat}} V/2$ is the total number of hard-core boson (cooper pair) modes.

Now the BCS ansatz can be expressed as

$$|\alpha\rangle = \prod_{\mathbf{k}; i=1\ldots N_{\text{flat}}}' \left[ 1 + e^{-\alpha} f^{(i)}(\mathbf{k}) \hat{\xi}_{\mathbf{k},i}^\dagger \hat{\xi}_{-\mathbf{k},i}^\dagger \right] |\text{vac}\rangle \tag{40}$$

and the normalization factor can be calculated as

$$\mathcal{N} = \langle \alpha|\alpha\rangle = \prod_{\mathbf{k}; i}' \left( 1 + e^{-2\Re\alpha} |f^{(I)}(\mathbf{k})|^2 \right) \tag{41}$$

which clearly admits the interpretation of the grand partition function of a classical hard-core boson gas. Now that the particle number and thus the filling fraction is only fixed on average but we expect the fluctuation to vanish in the thermodynamic limit:

$$\nu = \frac{1}{2\mathsf{N}_0 \mathcal{N}} \frac{\partial \mathcal{N}}{\partial (\Re\alpha)} = \frac{1}{\mathsf{N}_0} \sum_{\mathbf{k}; i}' \frac{e^{-2\Re\alpha}|f^{(I)}(\mathbf{k})|^2}{1 + e^{-2\Re\alpha}|f^{(I)}(\mathbf{k})|^2} \equiv \frac{1}{\mathsf{N}_0} \sum_{\mathbf{k}; i}' \nu^{(i)}(\mathbf{k}) \tag{42}$$

where we also introduced the filling fraction on each mode $\nu^{(i)}(\mathbf{k})$ for later convenience.

Now we are ready to evaluate the two-body reduced density matrix (2RDM) for this state, which is particularly convenient in the $\xi$ basis:

$$\Gamma_{IJKL} \equiv \langle \hat{\xi}_I^\dagger \hat{\xi}_J \hat{\xi}_K^\dagger \hat{\xi}_L \rangle / \mathcal{N} \tag{43}$$

where $I = (\mathbf{k}, i)$ is a compact index for electron modes.

Due to the special structure of these ansatz states, most of the entries of this tensor is zero. The only non-zero ones for the projected BCS states take the form ($\bar{I}$ means the paired mode of mode $I$):

$$\Gamma_{IIJJ} = \frac{1}{S_N^{\mathsf{N}_0}(\{|f^{(I)}|^2\})} \begin{cases} |f^{(I)}|^2 |f^{(J)}|^2 S_{N-2}^{\mathsf{N}_0-2}(\{|f|^2\} \backslash |f^{(I)}|^2, |f^{(J)}|^2) & I \neq J, \bar{J} \\ |f^{(I)}|^2 S_{N-1}^{\mathsf{N}_0-1}(\{|f|^2\} \backslash |f^{(I)}|^2) & \text{otherwise} \end{cases} \tag{44}$$

$$\Gamma_{IJJI} = \frac{1}{S_N^{\mathsf{N}_0}(\{|f^{(I)}|^2\})} \begin{cases} |f^{(I)}|^2 S_{N-1}^{\mathsf{N}_0-2}(\{|f|^2\} \backslash |f^{(I)}|^2, |f^{(J)}|^2) & I \neq J, \bar{J} \\ 0 & \text{otherwise} \end{cases} \tag{45}$$

$$\Gamma_{IJ\bar{I}\bar{J}} = \frac{1}{S_N^{\mathsf{N}_0}(\{|f^{(I)}|^2\})} \begin{cases} f^{(I),*} f^{(J)} S_{N-1}^{\mathsf{N}_0-2}(\{|f|^2\} \backslash |f^{(I)}|^2, |f^{(J)}|^2) & I \neq J, \bar{J} \\ |f^{(I)}|^2 S_{N-1}^{\mathsf{N}_0-1}(\{|f|^2\} \backslash |f^{(I)}|^2) & I = J \\ 0 & I = \bar{J} \end{cases} \tag{46}$$

where \ means excluding certain elements from a set.

For $N = 1$ (two-particle sector), the results can be simply evaluated as

$$\Gamma_{IIJJ} = \frac{1}{N_0} \begin{cases} 0 & I \neq J, \bar{J} \\ |f^{(I)}|^2/\bar{f}_{\text{sq}}^2 & \text{otherwise} \end{cases} \tag{47}$$

$$\Gamma_{IJJI} = \frac{1}{N_0} \begin{cases} |f^{(I)}|^2/\bar{f}_{\text{sq}}^2 & I \neq J, \bar{J} \\ 0 & \text{otherwise} \end{cases} \tag{48}$$

$$\Gamma_{IJ\bar{I}\bar{J}} = \frac{1}{N_0} \begin{cases} f^{(I),*} f^{(J)}/\bar{f}_{\text{sq}}^2 & I \neq J, \bar{J} \\ |f^{(I)}|^2/\bar{f}_{\text{sq}}^2 & I = J \\ 0 & I = \bar{J} \end{cases} \tag{49}$$

where

$$\bar{f}_{\text{sq}}^2 \equiv \frac{1}{2N_0} \sum_I |f^{(I)}|^2. \tag{50}$$

The 2RDM for the BCS states is even simpler:

$$\text{(Hartree)} \quad \Gamma_{IIJJ} = \begin{cases} \nu^{(I)} \nu^{(J)} & I \neq J, \bar{J} \\ \nu^{(I)} & \text{otherwise} \end{cases} \tag{51}$$

$$\text{(Fock)} \quad \Gamma_{IJJI} = \begin{cases} \nu^{(I)}(1 - \nu^{(J)}) & I \neq J, \bar{J} \\ 0 & \text{otherwise} \end{cases} \tag{52}$$

$$\text{(Cooper)} \quad \Gamma_{IJ\bar{I}\bar{J}} = \begin{cases} (1 - \nu^{(I)})(1 - \nu^{(J)})[f^{(I),*} f^{(J)} e^{-2\Re\alpha}] & I \neq J, \bar{J} \\ \nu^{(I)} & I = J \\ 0 & I = \bar{J} \end{cases} \tag{53}$$

We give the three types of terms different names according to the well-known convention for Wick's contraction of four fermion terms. These BCS states are indeed Slater determinant states of BdG fermions so they indeed satisfy the Wick's theorem.

### C. The variational energy

#### 1. General discussions

To derive the projectected Hamiltonian in the presence of $\bm{A}$, let's go back to the definition of the electronic structure:

$$\hat{H}_0 = \sum_{\bm{R}\bm{R}',\alpha\beta} t_{\alpha\beta,\bm{R}-\bm{R}'} e^{i\bm{A}\cdot(\bm{R}+\bm{r}_\alpha-\bm{R}'-\bm{r}_\beta)} \hat{c}_{\bm{R}\alpha}^\dagger \hat{c}_{\bm{R}'\beta} \tag{54}$$

where the flat connection is introduced by Periels substitution. The momentum space eigenbasis is still defined by $\hat{c}_{\bm{k}\alpha} \equiv \sum_{\bm{R}} e^{i\bm{k}\cdot(\bm{R}+\bm{r}_\alpha)} \hat{c}_{\bm{R},\alpha}/\sqrt{V}$, so the hopping matrix is modified to

$$t_{\alpha\beta}^\sigma(\bm{k}; \bm{A}) \equiv \sum_{\bm{R}} e^{i(\bm{k}+\bm{A})\cdot(\bm{R}+\bm{r}_\alpha-\bm{r}_\beta)} t_{\alpha\beta,\bm{R}} \tag{55}$$

$$= t_{\alpha\beta}^\sigma(\bm{k}+\bm{A}; \bm{A}=0) \tag{56}$$

and thus the diagonalized Hamiltonian reads

$$\hat{H}_0 = \sum_{\bm{k},n} \epsilon_n(\bm{k}+\bm{A}) \hat{\gamma}_{\bm{k}n}^\dagger \hat{\gamma}_{\bm{k}n}, \tag{57}$$

$$\hat{\gamma}_{\bm{k}n} = U_{n\alpha}^\dagger(\bm{k}+\bm{A}) \hat{c}_{\bm{k}\alpha}. \tag{58}$$

Similarly we have the substituted results for $\hat{S}_{\boldsymbol{R}}$

$$\hat{S}_{\boldsymbol{R}} = \sum_{\mu\nu,\boldsymbol{R}_1,\boldsymbol{R}_2} S_{\mu\nu}(\boldsymbol{R}_1,\boldsymbol{R}_2)e^{i\boldsymbol{A}\cdot(\boldsymbol{R}_1+\boldsymbol{r}_\mu-\boldsymbol{R}_2-\boldsymbol{r}_\nu)}\hat{c}^\dagger_{\boldsymbol{R}+\boldsymbol{R}_1,\mu}\hat{c}_{\boldsymbol{R}+\boldsymbol{R}_2,\nu} \tag{59}$$

$$= \frac{1}{\mathsf{V}} \sum_{\mu\nu,\boldsymbol{pq}} e^{i\boldsymbol{R}\cdot(\boldsymbol{p}-\boldsymbol{q})} S_{\mu\nu}(\boldsymbol{p}+\boldsymbol{A},\boldsymbol{q}+\boldsymbol{A})\hat{c}^\dagger_{\boldsymbol{p},\mu}\hat{c}_{\boldsymbol{q},\nu} \tag{60}$$

This means that the model Hamiltonian in the presence of $\boldsymbol{A}$ can be simply obtained as

$$\hat{H} = \frac{1}{\mathsf{V}} \sum_{\substack{\boldsymbol{pq},\boldsymbol{p'q'}\\nmkl\leq N_{\text{flat}}}} V_{nmkl}(\boldsymbol{p}+\boldsymbol{A},\boldsymbol{q}+\boldsymbol{A},\boldsymbol{p'}+\boldsymbol{A},\boldsymbol{q'}+\boldsymbol{A})\hat{\gamma}^\dagger_{\boldsymbol{p},n}\hat{\gamma}_{\boldsymbol{q},m}\hat{\gamma}^\dagger_{\boldsymbol{p'},k}\hat{\gamma}_{\boldsymbol{q'},l} \tag{61}$$

Now using the expression for 2RDM, we can evaluate the expectation value of the energy on the BCS ansatz state (only keeping the extensive part of the energy):

$$E_{\text{var}}(\boldsymbol{A}) = \frac{1}{\mathsf{V}}\sum_{\boldsymbol{p},\boldsymbol{q}} [E_{\text{Hartree},ij}(\boldsymbol{p},\boldsymbol{q}) + E_{\text{Fock},ij}(\boldsymbol{p},\boldsymbol{q}) + E_{\text{Cooper},ij}(\boldsymbol{p},\boldsymbol{q})] \tag{62}$$

These terms can be expressed in a succinct form:

$$E_{\text{Hartree}}(\boldsymbol{p},\boldsymbol{q}) = V(0)\operatorname{Tr}\left[\Xi^\dagger(\boldsymbol{p})\right]\operatorname{Tr}\left[\Xi(\boldsymbol{q})\right] \tag{63}$$

$$\Xi(\boldsymbol{q}) \equiv S(\boldsymbol{q}+\boldsymbol{A},\boldsymbol{q}+\boldsymbol{A})\nu_R(\boldsymbol{q}) \tag{64}$$

where we have made matrix generalization to the filling factor:

$$\nu_R(\boldsymbol{q}) \equiv e^{-2\Re\alpha}FF^\dagger(\boldsymbol{q})\left[\mathbb{1} + e^{-2\Re\alpha}FF^\dagger(\boldsymbol{q})\right]^{-1} \tag{65}$$

$$\nu_L(\boldsymbol{q}) \equiv e^{-2\Re\alpha}F^\dagger F(\boldsymbol{q})\left[\mathbb{1} + e^{-2\Re\alpha}F^\dagger F(\boldsymbol{q})\right]^{-1} \tag{66}$$

and the Fock and Cooper term can be combined as:

$$E_{\text{Fock+Cooper}}(\boldsymbol{p},\boldsymbol{q}) = V(\boldsymbol{q}-\boldsymbol{p})\sum_{n=0,m=0}^{\infty} e^{-2\Re\alpha}\operatorname{Tr}\left\{F^\dagger(\boldsymbol{p})\left[-e^{-2\Re\alpha}FF^\dagger(\boldsymbol{p})\right]^n S^\dagger(\boldsymbol{q}+\boldsymbol{A},\boldsymbol{p}+\boldsymbol{A})\left[-e^{-2\Re\alpha}FF^\dagger(\boldsymbol{q})\right]^m\right.$$
$$\left.\left[S(\boldsymbol{q}+\boldsymbol{A},\boldsymbol{p}+\boldsymbol{A})F(\boldsymbol{p})+F(\boldsymbol{q})S^T(-\boldsymbol{p}+\boldsymbol{A},-\boldsymbol{q}+\boldsymbol{A})\right]\right\} \tag{67}$$

$$= \frac{1}{2}V(\boldsymbol{q}-\boldsymbol{p})\sum_{n=0,m=0}^{\infty} e^{-2\Re\alpha}\operatorname{Tr}\left\{\left[-e^{-2\Re\alpha}F^\dagger F(\boldsymbol{p})\right]^n \left[S(\boldsymbol{q}+\boldsymbol{A},\boldsymbol{p}+\boldsymbol{A})F(\boldsymbol{p})+F(\boldsymbol{q})S^T(-\boldsymbol{p}+\boldsymbol{A},-\boldsymbol{q}+\boldsymbol{A})\right]^\dagger\right.$$
$$\left.\left[-e^{-2\Re\alpha}FF^\dagger(\boldsymbol{q})\right]^m \left[S(\boldsymbol{q}+\boldsymbol{A},\boldsymbol{p}+\boldsymbol{A})F(\boldsymbol{p})+F(\boldsymbol{q})S^T(-\boldsymbol{p}+\boldsymbol{A},-\boldsymbol{q}+\boldsymbol{A})\right]\right\} \tag{68}$$

$$= \frac{1}{2}V(\boldsymbol{q}-\boldsymbol{p})\operatorname{Tr}\left[\Upsilon(\boldsymbol{p},\boldsymbol{q})^\dagger \Upsilon(\boldsymbol{p},\boldsymbol{q})\right] \tag{69}$$

$$\Upsilon(\boldsymbol{p},\boldsymbol{q}) \equiv e^{-\Re\alpha}\left(\mathbb{1}-\nu_R\right)^{1/2}\left[S(\boldsymbol{q}+\boldsymbol{A},\boldsymbol{p}+\boldsymbol{A})F(\boldsymbol{p})+F(\boldsymbol{q})S^T(-\boldsymbol{p}+\boldsymbol{A},-\boldsymbol{q}+\boldsymbol{A})\right]\left(\mathbb{1}-\nu_L\right)^{1/2} \tag{70}$$

where in the second line we recombined the values at $(\boldsymbol{p},\boldsymbol{q})$ and $(-\boldsymbol{q},-\boldsymbol{p})$.

### 2. Sanity check

As a sanity check, let's plug in $F = F_0$ for the $\boldsymbol{A} = 0$ case. In this case the expressions can be greatly simplified

$$\operatorname{Tr}[\Xi_0(\boldsymbol{q})] = \sum_{m=1}^{\infty} \operatorname{Tr}\left\{A(\boldsymbol{q},-\boldsymbol{q})\left[-e^{-2\Re\alpha}F_0 F_0^\dagger(\boldsymbol{q})\right]^m F_0(\boldsymbol{q})\right\}$$

$$= -\sum_{m=1}^{\infty} \operatorname{Tr}\left\{A(-\boldsymbol{q},\boldsymbol{q})F_0(-\boldsymbol{q})\left[-e^{-2\Re\alpha}F_0^\dagger F_0(-\boldsymbol{q})\right]^m\right\} = -\operatorname{Tr}[\Xi_0(-\boldsymbol{q})] \tag{71}$$

$$\Upsilon_0(\boldsymbol{p},\boldsymbol{q}) \equiv e^{-\Re\alpha}\left(\mathbb{1}-\nu_{0,R}(\boldsymbol{q})\right)^{1/2}\left[A(\boldsymbol{q})A(\boldsymbol{q},-\boldsymbol{p})F(\boldsymbol{p})+F_0(\boldsymbol{q})A^T(-\boldsymbol{p},\boldsymbol{q})F_0^T(-\boldsymbol{p})\right]\left(\mathbb{1}-\nu_{0,L}(\boldsymbol{p})\right)^{1/2} = 0 \tag{72}$$

Note that we have used $F_0$ for the exact groundstate's form factor at $\boldsymbol{A} = 0$ (which is a part of the Hamiltonian). These equalities imply that the total energy of this state indeed sums to zero.

### D. Perturbative treatment

We now evaluate and optimize the variation energy to the second order in an expansion in small $\boldsymbol{A}$. We parametrize

$$F(\boldsymbol{k}) = F_0(\boldsymbol{k}) + \delta F(\boldsymbol{k}) \cdot \boldsymbol{A} + \ldots \tag{73}$$

$$S(\boldsymbol{p}+\boldsymbol{A}, \boldsymbol{q}+\boldsymbol{A}) = S(\boldsymbol{p}, \boldsymbol{q}) + \partial S(\boldsymbol{p}, \boldsymbol{q}) \cdot \boldsymbol{A} + \ldots \tag{74}$$

$$E = \boldsymbol{A} \cdot \delta^2 E \cdot \boldsymbol{A} + \ldots \tag{75}$$

note that $E$ should not have linear $\boldsymbol{A}$ dependence due to gauge invariance. The antisymmetric nature of form factor ensures that

$$\delta F(\boldsymbol{k}) = -\delta F^T(-\boldsymbol{k}). \tag{76}$$

We first evaluate the Hartree energy. We note two useful facts: 1. there should be no terms that have $\delta$ in only one of the traces, since all these terms always trivially vanish after momentum summation. 2. There are two identities:

$$\text{Tr}\left\{[F_0 F_0^\dagger(\boldsymbol{q})]^n \delta F(\boldsymbol{q}) F_0^\dagger(\boldsymbol{q}) [F_0 F_0^\dagger(\boldsymbol{q})]^m S(\boldsymbol{q}, \boldsymbol{q})\right\} = -\text{Tr}\left\{[F_0 F_0^\dagger(-\boldsymbol{q})]^{m+1} \delta F(-\boldsymbol{q}) F_0^\dagger(-\boldsymbol{q}) [F_0 F_0^\dagger(-\boldsymbol{q})]^{n-1} S(-\boldsymbol{q}, -\boldsymbol{q})\right\} \tag{77}$$

$$\text{Tr}\left\{[F_0 F_0^\dagger(\boldsymbol{q})]^n F_0(\boldsymbol{q}) \delta F^\dagger(\boldsymbol{q}) [F_0 F_0^\dagger(\boldsymbol{q})]^m S(\boldsymbol{q}, \boldsymbol{q})\right\} = -\text{Tr}\left\{[F_0 F_0^\dagger(-\boldsymbol{q})]^m F_0(-\boldsymbol{q}) \delta F^\dagger(-\boldsymbol{q}) [F_0 F_0^\dagger(-\boldsymbol{q})]^n S(-\boldsymbol{q}, -\boldsymbol{q})\right\} \tag{78}$$

which guarantee that most terms that contain $\delta F$ dependence will be canceled by their time-reversal partner terms. Then we reach the result:

$$\delta^2 E_{\text{Hartree}}(\boldsymbol{p}, \boldsymbol{q}) = V(0) \text{Tr}\left[\delta \Xi(\boldsymbol{p})^\dagger\right] \text{Tr}\left[\delta \Xi(\boldsymbol{q})\right] \tag{79}$$

$$\delta \Xi(\boldsymbol{q}) \equiv -\sum_{m=1}^{\infty} \partial S(\boldsymbol{q}, \boldsymbol{q}) \nu_{0,R}(\boldsymbol{q}) - e^{-2\Re \alpha} S(\boldsymbol{q}, \boldsymbol{q}) \delta F(\boldsymbol{q}) F_0^\dagger(\boldsymbol{q}) \left[\mathbb{1} - \nu_{0,R}(\boldsymbol{q})\right] \tag{80}$$

For QGN models the singular value spectrum of the ideal pairing form factor is degenerate and is the same for all momentum; that is

$$F_0(\boldsymbol{k}) F_0^\dagger(\boldsymbol{k}) = F_0^\dagger(\boldsymbol{k}) F_0(\boldsymbol{k}) = |f|^2 \mathbb{1} \implies \nu_{0,R}(\boldsymbol{k}) = \nu_{0,L}(\boldsymbol{k}) = (1 + e^{-2\Re \alpha} |f_0|^2)^{-1} \mathbb{1} \equiv \nu_0 \tag{81}$$

This actually implies that the perturbative Hartree energy now vanishes. To see this, we decompose

$$\delta \Xi(\boldsymbol{q}) \equiv \delta \Xi_1(\boldsymbol{q}) + \delta \Xi_2(\boldsymbol{q}) \tag{82}$$

$$\delta \Xi_1(\boldsymbol{q}) \equiv -\sum_{m=1}^{\infty} \partial S(\boldsymbol{q}, \boldsymbol{q}) \nu_0 \tag{83}$$

$$\delta \Xi_2(\boldsymbol{q}) \equiv -e^{-2\Re \alpha} S(\boldsymbol{q}, \boldsymbol{q}) \delta F(\boldsymbol{q}) F_0^\dagger \left[\mathbb{1} - \nu_0\right] \tag{84}$$

Since $\delta \Xi_1(\boldsymbol{q})$ is a total derivative, it sums to zero in momentum space. On the other hand,

$$\text{Tr}[\delta \Xi_2(\boldsymbol{q})] = -e^{-2\Re \alpha}(1-\nu_0) \text{Tr}\left\{S(\boldsymbol{q}, \boldsymbol{q}) \delta F(\boldsymbol{q}) F_0^\dagger\right\} \tag{85}$$

$$= e^{-2\Re \alpha}(1-\nu_0) \text{Tr}\left\{S^*(-\boldsymbol{q}, -\boldsymbol{q}) F_0^* \delta F^T(-\boldsymbol{q})\right\} \tag{86}$$

$$= e^{-2\Re \alpha}(1-\nu_0) \text{Tr}\left\{S(-\boldsymbol{q}, -\boldsymbol{q}) \delta F(-\boldsymbol{q}) F_0^\dagger\right\} = -\text{Tr}[\delta \Xi_2(-\boldsymbol{q})] \tag{87}$$

So it also sums to zero over $\boldsymbol{q}$ when evaluating $\delta^2 E_{\text{Hartree}}$. Thus we conclude

$$\delta^2 E_{\text{Hartree}} = 0 \tag{88}$$

for QGN models.

We next treat the Fock and Cooper parts together. Note that $\Upsilon(\boldsymbol{p}, \boldsymbol{q})$ start with linear in $\boldsymbol{A}$ term (cf. Eq. 72); therefore, we only need to keep the leading term therein:

$$\delta^2 E_{\text{Fock+Cooper}}(\boldsymbol{p}, \boldsymbol{q}) = \frac{1}{2} V(\boldsymbol{q}-\boldsymbol{p}) \text{Tr}\left[\delta \Upsilon(\boldsymbol{p}, \boldsymbol{q})^\dagger \delta \Upsilon(\boldsymbol{p}, \boldsymbol{q})\right] \tag{89}$$

$$\delta \Upsilon(\boldsymbol{p}, \boldsymbol{q}) \equiv e^{-\Re \alpha} \left(\mathbb{1} - \nu_{0,R}(\boldsymbol{q})\right)^{1/2} \left[\partial S(\boldsymbol{q}, \boldsymbol{p}) F_0(\boldsymbol{p}) + \delta F(\boldsymbol{q}) S^T(-\boldsymbol{p}, -\boldsymbol{q}) + S(\boldsymbol{q}, \boldsymbol{p}) \delta F(\boldsymbol{p}) \right.$$
$$\left. + F_0(\boldsymbol{q}) \partial S^T(-\boldsymbol{p}, -\boldsymbol{q})\right] \left(\mathbb{1} - \nu_{0,L}(\boldsymbol{p})\right)^{1/2} \tag{90}$$

Now that the $\nu_0$'s are now simply numbers, they can be pulled out of the trace in their expression

$$\delta^2 E_{\text{Fock+Cooper}}(\boldsymbol{p}, \boldsymbol{q}) = \frac{e^{-2\Re\alpha}(1-\nu_0)^2}{2} V(\boldsymbol{q}-\boldsymbol{p}) \text{Tr}\left[\delta\Upsilon'^\dagger(\boldsymbol{p},\boldsymbol{q})\delta\Upsilon'(\boldsymbol{p},\boldsymbol{q})\right] \tag{91}$$

$$\delta\Upsilon'(\boldsymbol{p},\boldsymbol{q}) \equiv \left[\partial S(\boldsymbol{q},\boldsymbol{p})F_0(\boldsymbol{p}) + \delta F(\boldsymbol{q})S^T(-\boldsymbol{p},-\boldsymbol{q}) + S(\boldsymbol{q},\boldsymbol{p})\delta F(\boldsymbol{p}) + F_0(\boldsymbol{q})\partial S^T(-\boldsymbol{p},-\boldsymbol{q})\right] \tag{92}$$

At an abstract level, if we regard the $\delta F$ on different momentum as a huge vector, the variational energy, before optimizing over all possible values of $\delta F$, takes the form:

$$\delta^2 E_{\text{var}}[\delta\vec{F}] \sim \left[\delta\vec{F}^\dagger \cdot M \cdot \delta\vec{F} + \vec{W}^\dagger \cdot \delta\vec{F} + \delta\vec{F}^\dagger \cdot \vec{W} + D\right] \tag{93}$$

where $M$ is a matrix, $\vec{W} \propto f$ is a vector, and $D \propto |f|^2$ is a constant. This is sufficient for us to infer that the optimized value in the parenthesis must be proportional to $|f|^2$. To be concrete, we define

$$\varepsilon \equiv \frac{1}{\mathsf{V}|f_0|^2} \min_{\delta F} \frac{1}{2\mathsf{V}} \sum_{\boldsymbol{pq}} V(\boldsymbol{q}-\boldsymbol{p})\text{Tr}\left[\delta\Upsilon'^\dagger(\boldsymbol{p},\boldsymbol{q})\delta\Upsilon'(\boldsymbol{p},\boldsymbol{q})\right] \tag{94}$$

Therefore, the optimized variational energy should be equal to:

$$\delta^2 E_{\text{var}} = \mathsf{V}\varepsilon(1-\nu_0)^2 e^{-2\Re\alpha}|f_0|^2 = \mathsf{V}\varepsilon(1-\nu_0)\nu_0 \tag{95}$$

The interesting point is that $\varepsilon$ is actually proportional to the inverse of the two-particle mass, $m_2^{-1}$. To see this, we use the 2RDM result for projected BCS in Eq. 47 to evaluate the energy as

$$\delta^2 E_{N=1} = \frac{1}{\mathsf{N}_0|f_0|^2} \sum_{\boldsymbol{pq}} \frac{1}{2\mathsf{V}} V(\boldsymbol{q}-\boldsymbol{p})\text{Tr}\left[\delta\Upsilon'^\dagger(\boldsymbol{p},\boldsymbol{q})\delta\Upsilon'(\boldsymbol{p},\boldsymbol{q})\right] \tag{96}$$

We skip the details of the derivation of this expression, which is almost identical to the case of BCS states, but would like to point out some key steps: 1) the Hartree terms trivially vanish in this sector. 2) in principle, the normalization factor is $|f_{\text{sq}}|^2$; however, since each $|F(\boldsymbol{q})|^2$ dependence in this quantity is $\propto 1/\mathsf{V}$, all the corrections due to $\delta F(\boldsymbol{q}) \cdot \boldsymbol{A}$ could be dropped. Then the optimized energy in this sector is

$$\delta^2 E_{N=1} = \frac{2}{\mathsf{N}_{\text{flat}}}\varepsilon \tag{97}$$

Since this variational ansatz already include all the states within this number and total momentum sector, this result is not only variational but actually *exact*: $E_{\text{var}} = E_{\text{actual}}$; this means that in the $N=1$ sector,

$$\frac{2}{\mathsf{N}_{\text{flat}}}\varepsilon = \delta^2 E_{\text{var}} = \delta^2 E_{\text{actual}} \equiv 4(m_2)^{-1}. \tag{98}$$

We thus reached our desired conclusion

$$\kappa_{\text{var}} \equiv \frac{1}{4\mathsf{V}}\delta^2 E_{\text{var}} = \frac{\mathsf{N}_{\text{flat}}}{2}(m_2)^{-1}(1-\nu_0)\nu_0 \succeq \kappa \tag{99}$$

Q.E.D.

---


\end{document}